\documentclass[journal]{IEEEtran}
\usepackage{setspace,cite}
\usepackage{latexsym,amsmath,amssymb}
\usepackage{graphicx}

\let\R\Real
\let\Z\Integer
\let\N\Natural

\def\map#1#2#3{#1 : #2 \rightarrow #3}
\def\deq{\stackrel{\scriptscriptstyle\triangle}{=}}

\def\argmin{\operatornamewithlimits{arg\,min}}
\def\E{\operatorname{{\mathbb E}}}
\def\Pr{\operatorname{{\mathbb P}}}

\def\cA{{\cal A}}

\def\cC{{\cal C}}
\def\cF{{\cal F}}

\def\cI{{\cal I}}

\def\cM{{\cal M}}

\def\cP{{\cal P}}
\def\cQ{{\cal Q}}
\def\cS{{\cal S}}

\def\cX{{\cal X}}
\def\cY{{\cal Y}}
\def\cZ{{\cal Z}}
\def\sS{{\sf S}}
\def\sV{{\sf V}}

\def\hcX{\widehat{\cal X}}
\def\hX{\widehat{X}}
\def\hx{\widehat{x}}

\def\W{{\mathbb W}}
\def\dvar{d}

\def\bd#1{\boldsymbol{#1}}
\def\wh#1{\widehat{#1}}
\def\td#1{\widetilde{#1}}

\newtheorem{theorem}{Theorem}[section]

\newtheorem{remark}{Remark}[section]
\newtheorem{lemma}{Lemma}[section]

\begin{document}

\title{Joint Universal Lossy Coding and\\
Identification of Stationary Mixing\\
Sources with General Alphabets}

\author{Maxim Raginsky,~\IEEEmembership{Member,~IEEE}%
\thanks{The material in this paper was presented in part at the IEEE International Symposium on Information Theory, Nice, France, June 2007. This work
    was supported by the Beckman Institute Fellowship.}%
  \thanks{M.~Raginsky was with the Beckman Institute for Advanced
    Science and Technology, University of
Illinois, Urbana, IL 61801 USA. He is now with the Department of Electrical and Computer Engineering, Duke University, Durham, NC 27708 USA (e-mail:~m.raginsky@duke.edu).} }%

\maketitle
\thispagestyle{empty}


\begin{abstract}
We consider the problem of joint universal variable-rate lossy coding and identification for parametric classes of stationary $\beta$-mixing sources with general (Polish) alphabets. Compression performance is measured in terms of Lagrangians, while identification performance is measured by the variational distance between the true source and the estimated source. Provided that the sources are mixing at a sufficiently fast rate and satisfy certain smoothness and Vapnik--Chervonenkis learnability conditions, it is shown that, for bounded metric distortions, there exist universal schemes for joint lossy compression and identification whose Lagrangian redundancies converge to zero as $\sqrt{V_n \log n /n}$ as the block length $n$ tends to infinity, where $V_n$ is the Vapnik--Chervonenkis dimension of a certain class of decision regions defined by the $n$-dimensional marginal distributions of the sources; furthermore, for each $n$, the decoder can identify $n$-dimensional marginal of the active source up to a ball of radius $O(\sqrt{V_n\log n/n})$ in variational distance, eventually with probability one. The results are supplemented by several examples of parametric sources satisfying the regularity conditions. 
\end{abstract}

\section{Introduction}
\label{sec:intro}

It is well known that lossless source coding and statistical modeling are complementary objectives. This fact is captured by the Kraft inequality (see Section~5.2 in Cover and Thomas \cite{CovTho91}), which provides a correspondence between uniquely decodable codes and probability distributions on a discrete alphabet. If one has full knowledge of the source statistics, then one can design an optimal lossless code for the source, and {\em vice versa}. However, in practice it is unreasonable to expect that the source statistics are known precisely, so one has to design {\em universal} schemes that perform asymptotically optimally within a given class of sources. In universal coding, too, as Rissanen has shown in \cite{Ris84,Ris96}, the coding and modeling objectives can be accomplished jointly: given a sufficiently regular parametric family of discrete-alphabet sources, the encoder can acquire the source statistics via maximum-likelihood estimation on a sufficiently long data sequence and use this knowledge to select an appropriate coding scheme. Even in nonparametric settings (e.g., the class of all stationary ergodic discrete-alphabet sources), universal schemes such as Ziv--Lempel \cite{ZivLem78} amount to constructing a probabilistic model for the source. In the reverse direction, Kieffer \cite{Kie93a} and Merhav \cite{Mer94}, among others, have addressed the problem of statistical modeling (parameter estimation or model identification) via universal lossless coding.

Once we consider {\em lossy} coding, though, the relationship between coding and modeling is no longer so simple. On the one hand, having full knowledge of the source statistics is certainly helpful for designing optimal rate-distortion codebooks. On the other hand, apart from some special cases (e.g., for i.i.d. Bernoulli sources and the Hamming distortion measure or for i.i.d. Gaussian sources and the squared-error distortion measure), it is not at all clear how to extract a reliable statistical model of the source from its reproduction via a rate-distortion code (although, as shown recently by Weissman and Ordentlich \cite{WeiOrd05}, the joint empirical distribution of the source realization and the corresponding codeword of a ``good" rate-distortion code converges to the distribution solving the rate-distortion problem for the source). This is not a problem when the emphasis is on compression, but there are situations in which one would like to compress the source and identify its statistics at the same time. For instance, in {\em indirect adaptive control} (see, e.g., Chapter~7 of Tao \cite{Tao03}) the parameters of the plant (the controlled system) are estimated on the basis of observation, and the controller is modified accordingly. Consider the discrete-time stochastic setting, in which the plant state sequence is a random process whose statistics are governed by a finite set of parameters. Suppose that the controller is geographically separated from the plant and connected to it via a noiseless digital channel whose capacity is $R$ bits per use. Then, given the time horizon $T$, the objective is to design an encoder and a decoder for the controller to obtain reliable estimates of both the plant parameters and the plant state sequence from the $2^{TR}$ possible outputs of the decoder.

To state the problem in general terms, consider an information source emitting a sequence $\bd{X} = \{X_i\}_{i \in \Z}$ of random variables taking values in an alphabet $\cX$. Suppose that the process distribution of $\bd{X}$ is not specified completely, but it is known to be a member of some parametric class $\{ P_\theta : \theta \in \Lambda \}$. We wish to answer the following two questions:
\begin{enumerate}
\item Is the class $\{P_\theta : \theta \in \Lambda\}$ universally encodable with respect to a given single-letter distortion measure $\rho$, by codes with a given structure (e.g., all fixed-rate block codes with a given per-letter rate, all variable-rate block codes, etc.)? In other words, does there exist a scheme that is asymptotically optimal for each $P_\theta$, $\theta \in \Lambda$?
\item If the answer to Question 1) is positive, can the codes be constructed in such a way that the decoder can not only reconstruct the source, but also identify its process distribution $P_\theta$, in an asymptotically optimal fashion?
\end{enumerate}
In previous work \cite{Rag05,Rag06}, we have addressed these two questions in the context of fixed-rate lossy block coding of stationary memoryless (i.i.d.) continuous-alphabet sources with parameter space $\Lambda$ a bounded subset of $\R^k$ for some finite $k$. We have shown that, under appropriate regularity conditions on the distortion measure and on the source models, there exist joint universal schemes for lossy coding and source identification whose redundancies (that is, the gap between the actual performance and the theoretical optimum given by the Shannon distortion-rate function) and source estimation fidelity both converge to zero as $O\big(\sqrt{\log n/n}\big)$, as the block length $n$ tends to infinity. The code operates by coding each block with the code matched to the source with the parameters estimated from the preceding block. Comparing this convergence rate to the $\log n/n$ convergence rate, which is optimal for redundancies of fixed-rate lossy block codes \cite{YanZha99}, we see that there is, in general, a price to be paid for doing compression and identification simultaneously. Furthermore, the constant hidden in the $O(\cdot)$ notation increases with the ``richness" of the model class $\{P_\theta : \theta \in \Lambda\}$, as measured by the Vapnik--Chervonenkis (VC) dimension \cite{VC71} of a certain class of measurable subsets of the source alphabet associated with the sources.

The main limitation of the results of \cite{Rag05,Rag06} is the i.i.d. assumption, which is rather restrictive as it excludes many practically relevant model classes (e.g., autoregressive sources, or Markov and hidden Markov processes). Furthermore, the assumption that the parameter space $\Lambda$ is bounded may not always hold, at least in the sense that we may not know the diameter of $\Lambda$ {\em a priori}. In this paper we relax both of these assumptions and study the existence and the performance of universal schemes for joint lossy coding and identification of stationary sources satisfying a mixing condition, when the sources are assumed to belong to a parametric model class $\{P_\theta : \theta \in \Lambda\}$, $\Lambda$ being an open subset of $\R^k$ for some finite $k$. Because the parameter space is not bounded, we have to use variable-rate codes with countably infinite codebooks, and the performance of the code is assessed by a composite Lagrangian functional \cite{ChoLooGra89} which captures the trade-off between the expected distortion and the expected rate of the code. Our result is that, under certain regularity conditions on the distortion measure and on the model class, there exist universal schemes for joint lossy source coding and identification such that, as the block length $n$ tends to infinity, the gap between the actual Lagrangian performance and the optimal Lagrangian performance achievable by variable-rate codes at that block length, as well as the source estimation fidelity at the decoder, converge to zero as $O\big( \sqrt{ V_n \log n / n}\big)$, where $V_n$ is the VC dimension of a certain class of decision regions induced by the collection $\{P^n_\theta : \theta \in \Lambda\}$ of the $n$-dimensional marginals of the source process distributions.

This result shows very clearly that the price to be paid for universality, in terms of both compression and identification, grows with the richness of the underlying model class, as captured by the VC dimension sequence $\{V_n\}$. The richer the model class, the harder it is to learn, which affects the compression performance of our scheme because we use the source parameters learned from past data to decide how to encode the current block. Furthermore, comparing the rate at which the Lagrangian redundancy decays to zero under our scheme with the $O(\log n/n)$ result of Chou, Effros and Gray \cite{ChoEffGra96}, whose universal scheme is not aimed at identification, we immediately see that, in ensuring to satisfy the twin objectives of compression and modeling, we inevitably sacrifice some compression performance.

The paper is organized as follows. Section~\ref{sec:prelims} introduces notation and basic concepts related to sources, codes and Vapnik--Chervonenkis classes. Section~\ref{sec:results} lists and discusses the regularity conditions that have to be satisfied by the source model class, and contains the statement of our result. The result is proved in Section~\ref{sec:proof}. Next, in Section~\ref{sec:examples} we give three examples of parametric source families (namely, i.i.d. Gaussian sources, Gaussian autoregressive sources and hidden Markov processes) which fit the framework of this paper under suitable regularity conditions. We conclude in Section~\ref{sec:conclusion} and outline directions for future research. Finally, the Appendix contains some technical results on Lagrange-optimal variable-rate quantizers.

\section{Preliminaries}
\label{sec:prelims}

\subsection{Sources}
\label{ssec:sources}

In this paper, a {\em source} is a discrete-time stationary ergodic random process $\bd{X} = \{X_i\}_{i \in \Z}$ with alphabet $\cX$. We assume that $\cX$ is a Polish space (i.e., a complete separable metric space\footnote{The canonical example is the Euclidean space $\R^d$ for some $d < \infty$.}) and equip $\cX$ with its Borel $\sigma$-field. For any pair of indices $i,j \in \Z$ with $i < j$, let $X^j_i$ denote the segment $(X_i,X_{i+1},\ldots,X_j)$ of $\bd{X}$.  If $P$ is the process distribution of $\bd{X}$, then we let $\E_P\{\cdot\}$ denote expectation with respect to $P$, and let $P^n$ denote the marginal distribution of $X^n_1$. Whenever $P$ carries a subscript, e.g., $P = P_\theta$, we write $\E_\theta\{\cdot\}$ instead. We assume that there exists a fixed $\sigma$-finite measure $\mu$ on $\cX$, such that the $n$-dimensional marginal of any process distribution of interest is absolutely continuous with respect to the product measure $\mu^n$, for all $n \ge 1$. We denote the corresponding densities $dP^n/d\mu^n$ by $p^n$. To avoid notational clutter, we omit the superscript $n$ from $\mu^n$, $P^n$ and $p^n$ whenever it is clear from the argument, as in $d\mu(x^n)$, $dP(x^n)$ or $p(x^n)$.

Given two probability measures $P,Q$ on a measurable space $(\cZ,\cA)$, the {\em variational distance} between them is defined by
$$
\dvar(P,Q) \deq \sup_{\{A_i\} \subseteq \cA} \sum_i |P(A_i) - Q(A_i)|,
$$
where the supremum is over all finite $\cA$-measurable partitions of $\cZ$ (see, e.g., Section~5.2 of Gray \cite{Gra90a}). If $p$ and $q$ are the densities of $P$ and $Q$, respectively, with respect to a dominating measure $\nu$, then we can write
$$
\dvar(P,Q) = \int_\cZ |p(z) - q(z)|d\nu(z).
$$
A useful property of the variational distance is that, for any measurable function $\map{f}{\cZ}{[0,1]}$, $|\E_Pf - \E_Qf| \le d(P,Q)$. When $P$ and $Q$ are $n$-dimensional marginals of $P_\theta$ and $P_{\theta'}$, respectively, i.e., $P = P^n_\theta$ and $Q = P^n_{\theta'}$, we write $d_n(\theta,\theta')$ for $\dvar(P^n_\theta,P^n_{\theta'})$. If $\cA'$ is a $\sigma$-subfield of $\cA$, we define the variational distance $\dvar(P,Q;\cA')$ between $P$ and $Q$ with respect to $\cA'$ by
$$
\dvar(P,Q;\cA') \deq \sup_{\{A_i\} \subseteq \cA'} \sum_i|P(A_i) - Q(A_i)|,
$$
where the supermum is over all finite $\cA'$-measurable partitions of $\cZ$. Given a $\delta > 0$ and a probability measure $P$, the {\em variational ball} of radius $\delta$ around $P$ is the set of all probability measures $Q$ with $\dvar(P,Q) \le \delta$.

Given a source $\bd{X}$ with process distribution $P$, let $P^0_{-\infty}$ and $P^\infty_1$ denote the marginal distributions of $P$ on $\{X_i\}_{i \le 0}$ and $\{ X_i\}_{i \ge 1}$, respectively. For each $k \ge 1$, the {\em $k$th-order absolute regularity coefficient} (or {\em $\beta$-mixing coefficient}) of $P$ is defined as \cite{VolRoz59,VolRoz61}:
$$
\beta_P(k) \deq \sup \left\{ \sum_i \sum_j |P(A_i \cap B_j) - P^0_{-\infty}(A_i)P^\infty_1(B_j)| \right\},
$$
where the supremum is over all finite $\sigma(X^0_{-\infty})$-measurable partitions $\{A_i\}$ and all finite $\sigma(X^\infty_k)$-measurable partitions $\{B_j\}$. Observe that
\begin{equation}
\beta_P(k) = \dvar\left(P,P^0_{-\infty} \times P^\infty_1; \sigma(X^0_{-\infty},X^\infty_k)\right),
\end{equation}
the variational distance between $P$ and the product distribution $P^0_{-\infty} \times P^\infty_1$ with respect to the $\sigma$-algebra $\sigma(X^0_{-\infty},X^\infty_k)$. Since $\bd{X}$ is stationary, we can ``split" its process distribution at any point $l \in \Z$ and define $\beta_P(k)$ equivalently by
\begin{equation}
\beta_P(k) \deq \dvar\left(P,P^l_{-\infty}\times P^\infty_{l+1}; \sigma(X^l_{-\infty},X^\infty_{l+k}) \right).
\end{equation}
Again, if $P$ is subscripted by some $\theta$, $P = P_\theta$, then we write $\beta_\theta(k)$.

\subsection{Codes}
\label{ssec:codes}

The class of codes we consider here is the collection of all finite-memory variable-rate vector quantizers. Let $\hcX$ be a {\em reproduction alphabet}, also assumed to be Polish. We assume that $\cX \cup \hcX$ is a subset of a Polish metric space $\cY$ with a bounded metric $\rho_0(\cdot,\cdot)$: there exists some $\rho_{\max} < +\infty$, such that $\rho_0(y,y') \le \rho_{\max}$ for all $y,y' \in \cY$. We take $\map{\rho}{\cX \times \hcX}{[0,\rho_{\max}]}$, $\rho(x,\hx) \deq \rho_0(x,\hx)$, as our (single-letter) {\em distortion function}.  A variable-rate vector quantizer with block length $n$ and memory length $m$ is a pair $C^{n,m} = (f,\varphi)$, where $\map{f}{\cX^n \times \cX^m}{\cS}$ is the {\em encoder}, $\map{\varphi}{\cS}{\hcX^n}$ is the {\em decoder}, and $\cS \subseteq \{0,1\}^*$ is a countable collection of binary strings satisfying the prefix condition or, equivalently, the Kraft inequality
$$
\sum_{s \in \cS} 2^{-\ell(s)} \le 1,
$$
where $\ell(s)$ denotes the length of $s$ in bits. The mapping of the source $\bd{X}$ into the reproduction process $\bd{\hX}$ is defined by
$$
\hX^{n(k+1)}_{nk + 1} = \varphi\left(f\big(X^{n(k+1)}_{nk + 1},X^{nk}_{nk - m + 1}\big)\right), \qquad k \in \Z.
$$
That is, the encoding is done in blocks of length $n$, but the encoder is also allowed to observe the $m$ symbols immediately preceding each block. The {\em effective memory} of $C^{n,m}$ is defined as the set $\cM \subseteq \{1,\ldots,m\}$, such that
$$
f(x^m) = f(\td{x}^m), \qquad \forall x^m,\td{x}^m \in \cX^m : x_i = \td{x}_i, \forall i \in \cM.
$$
The size $|\cM|$ of $\cM$ is called the {\em effective memory length} of $C^{n,m}$. We shall often use $C^{n,m}$ to also denote the composite mapping $\varphi \circ f$: $\hX^n_1 = C^{n,m}(X^n_1,X^0_{-m+1})$. When the code has zero memory ($m=0$), we shall denote it more compactly by $C^n$.

The performance of the code on the source with process distribution $P$ is measured by its expected distortion
$$
D_P(C^{n,m}) \deq \E_P\left\{ \rho_n(X^n_1,\hX^n_1) \right\},
$$
where for $x^n \in \cX^n$ and $\hx^n \in \hcX^n$, $\rho_n(x^n,\hx^n) \deq n^{-1}\sum^n_{i=1}\rho(x_i,\hx_i)$ is the per-letter distortion incurred in reproducing $x^n$ by $\hx^n$, and by its expected rate
$$
R_P(C^{n,m}) \deq \E_P\left\{ \ell_n\left(f\left(X^n_1,X^0_{-m+1}\right)\right) \right\},
$$
where $\ell_n(s)$ denotes the length of a binary string $s$ in bits, normalized by $n$. (We follow Neuhoff and Gilbert \cite{NeuGil82} and normalize the distortion and the rate by the length $n$ of the {\em reproduction} block, not by the combined length $n+m$ of the source block plus the memory input.) When working with variable-rate quantizers, it is convenient \cite{ChoLooGra89,Lin01} to absorb the distortion and the rate into a single performance measure, the {\em Lagrangian distortion}
$$
L_P(C^{n,m},\lambda) \deq D_P(C^{n,m}) + \lambda R_P(C^{n,m}),
$$
where $\lambda > 0$ is the {\em Lagrange multiplier} which controls
the distortion-rate trade-off. Geometrically, $L_P(C^{n,m})$ is the $y$-intercept of the line with slope $-\lambda$, passing through the point $\left(R_P(C^{n,m}),D_P(C^{n,m})\right)$ in the rate-distortion plane \cite{EffChoGra94}. If $P$ carries a subscript, $P = P_\theta$, then we write $D_\theta(\cdot)$, $R_\theta(\cdot)$ and $L_\theta(\cdot)$.

\subsection{Vapnik--Chervonenkis classes}
\label{ssec:vc}

In this paper, we make heavy use of Vapnik--Chervonenkis theory (see Devroye, Gy\"orfi and Lugosi \cite{DevGyoLug96}, Vapnik \cite{Vap98}, Devroye and Lugosi \cite{DevLug01} or Vidyasagar \cite{Vid03} for detailed treatments). This section contains a brief summary of the needed concepts and results. Let $(\cZ,\cA)$ be a measurable space. For any collection $\cC
\subseteq \cA$ of measurable subsets of $\cZ$ and any $n$-tuple $z^n \in \cZ^n$, define the set $\cC(z^n)
\subseteq \{0,1\}^n$ consisting of all distinct binary strings of the
form $(1_{\{z_1 \in A\}},\ldots,1_{\{z_n \in A\}})$, $A \in \cC$. Then
$$
\sS_n(\cC) \deq \max_{z^n \in \cZ^n} |\cC(z^n)|
$$
is called the {\em $n$th shatter coefficient} of $\cC$. The {\em
  Vapnik--Chervonenkis dimension} (or VC-dimension) of $\cC$, denoted
by $\sV(\cC)$, is
defined as the largest $n$ for which $\sS_n(\cC) = 2^n$ (if $\sS_n(\cC)
= 2^n$ for all $n=1,2,\ldots$, then we set $\sV(\cC) = \infty$). If
$\sV(\cC) < \infty$, then $\cC$ is called a {\em Vapnik--Chervonenkis
  class} (or VC class). If $\cC$ is a VC class with $\sV(\cC) \ge 2$,
then it follows from the results of Vapnik and Chervonenkis
\cite{VC71} and Sauer \cite{Sau72} that $\sS_n(\cC) \le n^{\sV(\cC)}$.

For a VC class $\cC$, the so-called {\em Vapnik--Chervonenkis inequalities} (see Lemma~\ref{lm:vc} below) relate its VC dimension $\sV(\cC)$ to maximal deviations of the probabilities
of the events in $\cC$ from
their relative frequencies with respect to an i.i.d. sample of size $n$. For any $z^n \in \cZ^n$, let
$$
P_{z^n} = \frac{1}{n}\sum^n_{i=1}\delta_{z_i}
$$
denote the induced empirical distribution, where
$\delta_{z_i}$ is the Dirac measure (point mass) concentrated at
$z_i$. We then have the following:

\begin{lemma}[Vapnik--Chervonenkis inequalities]
\label{lm:vc}

Let $P$ be a probability measure on $(\cZ,\cA)$, and $Z^n_1 =
(Z_1,\ldots,Z_n)$ an $n$-tuple of independent random variables with
$Z_i \sim P$, $1 \le i \le n$. Let $\cC$ be a Vapnik--Chervonenkis
class with $\sV(\cC) \ge 2$. Then for every $\delta > 0$,
\begin{equation}
\Pr\left\{ \sup_{A \in \cC} |P_{Z^n_1}(A) - P(A)| > \delta \right\} \le
8n^{\sV(\cC)}e^{-n\delta^2/32}
\label{eq:vcbound_1}
\end{equation}
and
\begin{equation}
\E\left\{ \sup_{A \in \cC} |P_{Z^n_1}(A) - P(A)| \right\} \le
c\sqrt{\frac{\sV(\cC)\log n}{n}},
\label{eq:vcbound_2}
\end{equation}
where $c > 0$ is a universal constant. The probabilities and
expectations are with respect to the product measure $P^n$ on $(\cZ^n,\cA^n)$.
\end{lemma}

\begin{remark} A more refined technique involving metric
entropies and empirical covering numbers, due to Dudley \cite{Dud78}, can yield a much better
$O(1/\sqrt{n})$ bound on the expected maximal deviation between the
true and the empirical probabilities. This improvement, however, comes at
the expense of a much larger constant hidden in the $O(\cdot)$
notation.\end{remark}

Finally, we shall need the following lemma, which is a simple corollary of the results of Karpinski and Macintyre \cite{KarMac97} (see also Section~10.3.5 of Vidyasagar \cite{Vid03}):

\begin{lemma} \label{lm:karpinski_macintyre}
Let $\cC = \{A_\xi : \xi \in \R^k\}$ be a collection of measurable subsets of $\R^d$, such that
$$
A_\xi = \{ z \in \R^d : \Pi(z,\xi) > 0 \}, \qquad \xi \in \R^k
$$
where for each $z \in \R^d$, $\Pi(z,\cdot)$ is a polynomial of degree $s$ in the components of $\xi$. Then $\cC$ is a VC class with $\sV(\cC) \le 2k\log(4es)$.
\end{lemma}

\section{Statement of results}
\label{sec:results}

In this section we state our result concerning universal schemes for joint lossy compression and identification of stationary sources under certain regularity conditions. We work in the usual setting of universal source coding: we are given a source $\bd{X} = \{X_i\}_{i \in \Z}$ whose process distribution is known to be a member of some parametric class $\{P_\theta : \theta \in \Lambda\}$. The parameter space $\Lambda$ is an open subset of the Euclidean space $\R^k$ for some finite $k$, and we assume that $\Lambda$ has nonempty interior. We wish to design a sequence of variable-rate vector quantizers, such that the decoder can reliably reconstruct the original source sequence $\bd{X}$ and reliably identify the active source in an asymptotically optimal manner for all $\theta \in \Lambda$. We begin by listing the regularity conditions.

{\em Condition 1.} The sources in $\{P_\theta : \theta \in
\Lambda\}$ are {\em algebraically $\beta$-mixing}: there exists a constant $r > 0$, such that
$$
\beta_\theta(k) = O(k^{-r}), \qquad \forall \theta \in \Lambda
$$
where the constant implicit in the $O(\cdot)$ notation may depend on $\theta$.

This condition ensures that certain finite-block functions of the source $\bd{X}$ can be approximated in distribution by i.i.d. processes, so that we can invoke the Vapnik--Chervonenkis machinery of Section~\ref{ssec:vc}. This ``blocking" technique, which we exploit in the proof of our Theorem~\ref{thm:main}, dates back to Bernstein \cite{Ber27}, and was used by Yu \cite{Yu94} to derive rates of convergence in the uniform laws of large numbers for stationary mixing processes, and by Meir \cite{Mei00} in the context of nonparametric adaptive prediction of stationary time series. As an example of when an even stronger decay condition holds, let
$\bd{X} = \{X_i\}_{i \in \Z}$ be a finite-order autoregressive moving-average
  (ARMA) process driven by a zero-mean i.i.d. process $\bd{Y} = \{Y_i\}$,
i.e., there exist poisitive integers $p,q$ and $p+q+1$ real constants
$a_0,a_1\ldots,a_p,b_1,\ldots,b_q$ such that
$$
\sum^p_{i=0} a_i X_{n-i} = \sum^q_{j=1} b_j Y_{n-j}, \qquad n \in \Z.
$$
Mokkadem \cite{Mok88} has shown that, provided the common distribution of the $Y_i$ is absolutely continuous and the roots of the
polynomial $A(z) = \sum^p_{i=0} a_i z^i$ lie outside the unit
circle in the complex plane, the $\beta$-mixing coefficients of $\bd{X}$ decay to zero exponentially.

{\em Condition 2.} For each $\theta \in \Lambda$, there exist
constants $\delta_\theta > 0$ and $c_\theta > 0$, such that
$$
\sup_n  \frac{d_n(\theta,\theta')}{\sqrt{n}} \le c_\theta \|\theta - \theta'\|
$$
for all $\theta'$ in the open ball of radius $\delta_\theta$ centered at $\theta$, where $\|\cdot\|$ is the Euclidean norm on $\Lambda$.

This condition guarantees that, for any sequence $\{\delta_n\}_{n \in \N}$ of positive reals such that
$$
\delta_n \to 0, \sqrt{n}\delta_n \to 0, \qquad \mbox{as } n \to \infty
$$
and any sequence $\{\theta_n\}_{n \in \N}$ in $\Lambda$ satisfying $\| \theta_n - \theta \| < \delta_n$ for a given $\theta \in \Lambda$, we have
$$
d_n(\theta,\theta_n) \to 0, \qquad \mbox{as } n \to \infty.
$$
It is weaker (i.e., more general) than the conditions of Rissanen
\cite{Ris84,Ris96} which
control the behavior of the relative entropy (information divergence) as a function of the source parameters in terms of the Fisher information and related quantities. Indeed, for each $n$ let
\begin{eqnarray*}
D_n(P_\theta \| P_{\theta'}) &=& \frac{1}{n} \E_\theta \left\{ \ln \frac{dP_\theta}{dP_{\theta'}}(X^n_1)\right\} \\
&\equiv& \frac{1}{n} \int_{\cX^n} p_\theta(x^n) \ln \frac{p_\theta(x^n)}{p_{\theta'}(x^n)} d\mu(x^n)
\end{eqnarray*}
be the normalized $n$th-order relative entropy (information
divergence) between $P_\theta$ and $P_{\theta'}$. Suppose that, for
each $\theta$,
$D_n(P_\theta \| P_{\theta'})$ is twice continuously differentiable as
a function of $\theta'$. Let $\theta'$ lie in an open ball of radius
$\delta$ around $\theta$. Since $D(P_\theta \|
P_{\theta'})$ attains its minimum at $\theta' = \theta$, the gradient $\nabla_{\theta'} D_n(P_\theta \| P_{\theta'})$ evaluated at $\theta' = \theta$ is zero, and we can write the second-order Taylor expansion of $D_n$ about $\theta$ as
\begin{equation}
D_n(P_\theta \| P_{\theta'}) = \frac{1}{2}(\theta - \theta')^T
J_n(\theta)(\theta - \theta') + o(\| \theta - \theta' \|^2),
\label{eq:taylor}
\end{equation}
where the Hessian matrix
$$
J_n(\theta) = \left[ \frac{\partial^2}{\partial \theta'_i \partial \theta'_j} D_n(P_\theta\|P_{\theta'}) \Big|_{\theta' = \theta} \right]_{i,j=1,\ldots,k},
$$
under additional regularity conditions, is equal to the Fisher information matrix
$$
I_n(\theta) = \left[ -\frac{1}{n}\E_\theta \left\{ \frac{\partial^2}{\partial \theta_i \partial \theta_j} \ln p_\theta(X^n_1) \right\}\right]_{i,j = 1,\ldots,k}
$$
(see Clarke and Barron \cite{ClaBar90}). Assume now, following Rissanen \cite{Ris84,Ris96}, that the sequence of matrix
norms $\{\|I_n(\theta)\|\}$ is bounded (by a constant
depending on $\theta$). Then we can write
\begin{eqnarray*}
D_n(P_\theta \| P_{\theta'}) &\le& \frac{1}{2} (\| I_n(\theta) \|+ o(1)) \cdot \| \theta - \theta'
  \|^2 \\
  &\le& c'_\theta \| \theta - \theta'\|^2,
\end{eqnarray*}
i.e., the normalized relative entropies $D_n(P_\theta \| P_{\theta'})$
are locally quadratic in $\theta'$. Then Pinsker's inequality (see, e.g., Lemma~5.2.8 of Gray
\cite{Gra90a}) implies that $\sqrt{2 D_n(P_\theta \| P_{\theta'})} \ge d_n(\theta,\theta')/\sqrt{n}$, and we recover our
Condition 2. Rissanen's condition, while stronger than our Condition 2, is easier to check, the fact which we exploit in our discussion of examples of Section~\ref{sec:examples}.

{\em Condition 3.} 
For each $n$, let $\cA_n$ be the collection of all sets of the form
$$
A_{\theta,\theta'} = \big\{ x^n \in \cX^n: p_\theta(x^n) >
p_{\theta'}(x^n) \big\}, \qquad \theta \neq \theta'.
$$
Then we require that, for each $n$, $\cA_n$ is a VC class, and
$\sV(\cA_n) = o(n/\log n)$.

This condition is satisfied, for example, when $\sV(\cA_n) = V < \infty$ independently of $n$, or when $\sV(\cA_n) = \log n$. The use of the class $\cA_n$ dates back to the work of Yatracos \cite{Yat85} on minimum-distance density estimation. The ideas of Yatracos were further developed by Devroye and Lugosi \cite{DevLug96,DevLug97}, who dubbed $\cA_n$ the {\em Yatracos class} (associated with the densities $p^n_\theta$). We shall adhere to this terminology. To give an intuitive interpretation to $\cA_n$, let us consider a pair $\theta,\theta' \in \Lambda$ of distinct parameter vectors and note that the set $\{x^n : p_\theta(x^n) > p_{\theta'}(x^n)\}$ consists of all $x^n$ for which the simple hypothesis test
\begin{equation}
H_0 : X^n_1 \sim P^n_\theta \quad \mbox{versus} \quad H_1 : X^n_1 \sim
P^n_{\theta'}
\label{eq:hyptest}
\end{equation}
is passed by the null hypothesis $H_0$ under the likelihood-ratio decision rule. Now, suppose that $Z_1,\ldots,Z_m$ are drawn {\em independently} from $P^n_\theta$. To each $A \in \cA_n$ we can associate a {\em classifier} $\map{\kappa_A}{\cX^n}{\{0,1\}}$ defined by $\kappa_A(x^n) \deq 1_{\{x^n \in A\}}$. Call two sets $A,A' \in \cA_n$ {\em equivalent} with respect to the sample $Z^n_1 = (Z_1,\ldots,Z_m)$, and write $A \sim_{Z^n_1} A'$, if their associated classifiers yield identical classification patterns:
$$
\big(\kappa_A(Z_1),\ldots,\kappa_A(Z_m)\big) = \big(\kappa_{A'}(Z_1),\ldots,\kappa_{A'}(Z_m)\big).
$$
It is easy to see that $\sim_{Z^n_1}$ is an equivalence relation. From the definitions of the shatter coefficients $\sS_m(\cA_n)$ and the VC dimension $\sV(\cA_n)$ (cf.~Section~\ref{ssec:vc}), we see that the cardinality of the quotient set $\cA_n/\sim_{Z^n_1}$ is equal to $2^m$ for all sample sizes $m \le \sV(\cA_n)$, whereas for $m > \sV(\cA_n)$, it is bounded from above by $m^{\sV(\cA_n)}$, which is strictly less than $2^m$. Thus, the fact that the Yatracos class $\cA_n$ has finite VC dimension implies that the problem of estimating the density $p^n_\theta$ from a large i.i.d. sample reduces, in a sense, to a finite number (in fact, polynomial in the sample size $m$, for $m > \sV(\cA_n)$) of simple hypothesis tests of the type (\ref{eq:hyptest}). Our Condition 1 will then allow us to transfer this intuition to (weakly) dependent samples.

Now that we have listed the regularity conditions that must hold for the sources in $\{P_\theta : \theta \in \Lambda\}$, we can state our main result.

\begin{theorem} \label{thm:main}
Let $\{ P_\theta : \theta \in \Lambda \}$ be a parametric class of sources satisfying Conditions 1--3. Then for every $\lambda >0$ and every $\eta > 0$, there exists a sequence $\{C^{n,m_n}_*\}_{n \in \N}$ of variable-rate vector quantizers with memory length $m_n \le n(n + n^{(2+\eta)/r} + 1)$ and effective memory length $n^2$, such that, for all $\theta \in \Lambda$,
\begin{eqnarray}
&& L_\theta(C^{n,m_n}_*,\lambda) - \inf_{m \ge 0} \inf_{C^{n,m}}
L_\theta(C^{n,m},\lambda) \nonumber\\
&& \qquad \qquad \qquad \qquad = O\left(\sqrt{\frac{\sV(\cA_n)\log n}{n}}\right),\label{eq:universality}
\end{eqnarray}
where the constants implicit in the $O(\cdot)$ notation depend on $\theta$. Furthermore, for each $n$, the binary description produced by the encoder is such that the decoder can identify the $n$-dimensional marginal of the active source up to a variational ball of radius $O\big(\sqrt{\sV(\cA_n)\log n/n}\big)$ with probability one.\\
\end{theorem}

What (\ref{eq:universality}) says is that, for each block length $n$ and each $\theta \in \Lambda$, the code $C^{n,m_n}_*$, which is {\em independent of $\theta$}, performs almost as well as the best finite-memory quantizer with block length $n$ that can be designed with full {\em a priori} knowledge of the $n$-dimensional marginal $P^n_\theta$. Thus, as far as compression goes, our scheme can compete with all finite-memory variable-rate lossy block codes (vector quantizers), with the additional bonus of allowing the decoder to identify the active source in an asymptotically optimal manner.

It is not hard to see that the double infimum in (\ref{eq:universality}) is achieved already in the zero-memory case, $m = 0$. Indeed, it is immediate that having nonzero memory can only improve the Lagrangian performance, i.e.,
$$
\inf_{m \ge 0} \inf_{C^{n,m}} L_\theta(C^{n,m},\lambda) \le \inf_{C^n} L_\theta(C^n,\lambda), \qquad \forall \theta \in \Lambda.
$$
On the other hand, given any code $C^{n,m} = (f,\varphi)$, we can construct a zero-memory code $C^n_0 = (f_0,\varphi_0)$, such that $L_\theta(C^n_0,\lambda) \le L_\theta(C^{n,m},\lambda)$ for all $\theta \in \Lambda$. To see this, define for each $x^n \in \cX^n$ the set
$$
\cS(x^n) \deq \{s \in \{0,1\}^* : s = f(x^n,z^m) \mbox{ for some } z^m \in \cX^m \},$$
and let
$$
f_0(x^n) = \argmin_{s \in \cS(x^n)} \big( \rho_n(x^n,\varphi(s)) + \lambda \ell(s) \big), \qquad \forall x^n \in \cX^n
$$
and $\varphi_0 \equiv \varphi$. Then, given any $(x^n,z^m) \in \cX^n\times \cX^m$, let $s = f(x^n,z^m)$. We then have
\begin{eqnarray*}
&& \rho_n(x^n,C^n_0(x^n)) + \lambda \ell(f_0(x^n)) \\
&& \qquad = \rho_n(x^n,\varphi(f_0(x^n)) + \lambda\ell(f_0(x^n)) \\
&& \qquad \le \rho_n(x^n,\varphi(s)) + \ell(s) \\
&& \qquad = \rho_n(x^n,f(x^n,z^m)) + \ell(f(x^n,z^m)).
\end{eqnarray*}
Taking expectations, we see that $L_\theta(C^n_0,\lambda) \le L_\theta(C^{n,m},\lambda)$ for all $\theta \in \Lambda$, which proves that
$$
\inf_{C^n} L_\theta(C^n,\lambda) \le \inf_{m \ge 0} \inf_{C^{n,m}} L_\theta(C^{n,m},\lambda), \qquad \forall \theta \in \Lambda.
$$
The infimum of $L_\theta(C^n,\lambda)$ over all zero-memory variable-rate quantizers $C^n$ with block length $n$ is the {\em operational $n$th-order distortion-rate Lagrangian} $\wh{L}^n_\theta(\lambda)$ \cite{EffChoGra94}. Because each $P_\theta$ is ergodic, $\wh{L}^n_\theta(\lambda)$ converges to the {\em distortion-rate Lagrangian}
$$
L_\theta(\lambda) \deq \min_R \big( D_\theta(R) + \lambda R \big),
$$
where $D_\theta(R)$ is the Shannon distortion-rate function of $P_\theta$ (see Lemma~2 in the Appendix to Chou, Effros and Gray \cite{ChoEffGra96}). Thus, our scheme is universal not only in the $n$th-order sense of (\ref{eq:universality}), but also in the distortion-rate sense, i.e.,
$$
L_\theta(C^{n,m_n}_*,\lambda) - L_\theta(\lambda) \to 0, \qquad \mbox{as } n \to \infty
$$
for every $\theta \in \Lambda$. Thus, in the terminology of \cite{ChoEffGra96}, our scheme is {\em weakly minimax universal} for $\{P_\theta : \theta \in \Lambda\}$.

\section{Proof of Theorem~\ref{thm:main}}
\label{sec:proof}

\subsection{The main idea}
\label{ssec:main}

In this section, we describe the main idea behind the proof and fix some notation. We have already seen that it suffices to construct a universal scheme that can compete with all {\em zero-memory} variable-rate quantizers. That is, it suffices to show that there exists a sequence $\{C^{n,m_n}_*\}$ of codes, such that
\begin{equation} \label{eq:zeromem_universality}
L_\theta(C^{n,m_n}_*,\lambda) - \wh{L}^n_\theta(\lambda) = O\left(\sqrt{\frac{\sV(\cA_n)\log n}{n}}\right),  \forall \theta \in \Lambda.
\end{equation}
This is what we shall prove.

We assume throughout that the ``true" source is $P_{\theta_0}$ for some $\theta_0 \in \Lambda$. Our code operates as follows. Suppose that:
\begin{itemize}
\item Both the encoder and the decoder have access to a countably infinite ``database" $\bd{c} = \{\theta(i)\}_{i \in \N}$, where each $\theta(i) \in \Lambda$. Using Elias' universal representation of the integers \cite{Eli75}, we can associate to each $\theta(i)$ a unique binary string $s(i)$ with $\ell(s(i)) = \log i + O(\log \log i)$ bits.
\item A sequence $\{ \delta_n \}$ of positive reals is given, such that
$$
\delta_n \to 0, \sqrt{n}\delta_n \to 0, \qquad \mbox{as } n \to \infty
$$
(we shall specify the sequence $\{\delta_n\}$ later in the proof).
\item For each $n \in \N$ and each $\theta \in \Lambda$, there exists a zero-memory $n$-block code $C^n_\theta = (f^n_\theta,\varphi^n_\theta)$ that achieves (or comes arbitrarily close to) the $n$th-order Lagrangian optimum for $P_\theta$: $L_\theta(C^n_\theta,\lambda) = \wh{L}^n_\theta(\lambda)$.
\end{itemize}
Fix the block length $n$. Because the source is stationary, it suffices to describe the mapping of $X^n_1$ into $\hX^n_1$. The encoding is done as follows:
\begin{enumerate}
\item The encoder estimates $P^n_{\theta_0}$ from the $m_n$-block $X^0_{-m_n+1}$ as $P^n_{\td{\theta}}$, where $\td{\theta} = \td{\theta}(X^0_{-m_n+1})$.
\item The encoder then computes the {\em waiting time}
$$
T_n \deq \inf \left\{ i \ge 1 : d_n\big(\theta(i),\td{\theta}(X^0_{-m_n+1})\big) \le \sqrt{n}\delta_n \right\},
$$
with the standard convention that the infimum of the empty set is equal to $+\infty$. That is, the encoder looks through the database $\bd{c}$ and finds the first $\theta(i)$, such that the $n$-dimensional distribution $P^n_{\theta(i)}$ is in the variational ball of radius $\sqrt{n}\delta_n$ around $P^n_{\td{\theta}}$. 
\item If $T_n < + \infty$, the encoder sets $\wh{\theta} = \theta(i)$; otherwise, the encoder sets $\wh{\theta} = \theta_d$, where $\theta_d \in \Lambda$ is some default parameter vector, say, $\theta(1)$.
\item The binary description of $X^n_1$ is a concatenation of the following three binary strings: (i) a 1-bit flag $b$ to tell whether $T_n$ is finite $(b=0)$ or infinite $(b=1)$; (ii) a binary string $s_1$ which is equal to $s(T_n)$ if $T_n < +\infty$ or to an empty string if $T_n = +\infty$; (iii) $s_2 = f_{\wh{\theta}}(X^n_1)$. The string $\td{s} = bs_1$ is called the {\em first-stage description}, while $s_2$ is called the {\em second-stage description}.
\end{enumerate}
The decoder receives $bs_1s_2$, determines $\wh{\theta}$ from $\td{s}$, and produces the reproduction $\wh{X}^n_1 = \varphi_{\wh{\theta}}(s_2)$. Note that when $b = 0$ (which, as we shall show, will happen eventually almost surely), $P^n_{\wh{\theta}}$ lies in the variational ball of radius $\sqrt{n}\delta_n$ around the estimated source $P^n_{\td{\theta}}$. If the latter is a good estimate of $P^n_{\theta_0}$, i.e., $d_n(\theta_0,\td{\theta}) \to 0$ as $n \to \infty$ a.s., then the estimate of the true source computed by the decoder is only slightly worse. Furthermore, as we shall show, the almost-sure convergence of $d_n(\theta_0,\wh{\theta})$ to zero as $n \to \infty$ implies that the Lagrangian performance of $C^n_{\wh{\theta}}$ on $P_{\theta_0}$ is close to the optimum $L_{\theta_0}(C^n_{\theta_0},\lambda) \equiv \wh{L}^n_{\theta_0}(\lambda)$.

Formally, the code $C^{n,m_n}_*$ is comprised by the following maps:
\begin{itemize}
\item the {\em parameter estimator} $\map{\td{\theta}}{\cX^{m_n}}{\Lambda}$;
\item the {\em parameter encoder} $\map{\td{g}}{\Lambda}{\td{\cS}}$, where $\td{\cS} \equiv \{ 0 s(i) \}_{i \in \N} \cup \{1 \}$;
\item the {\em parameter decoder} $\map{\td{\psi}}{\td{\cS}}{\Lambda}$.
\end{itemize}
Let $\td{f}$ denote the composition $\td{g}\circ\td{\theta}$ of the parameter estimator and the parameter encoder, which we refer to as the {\em first-stage encoder}, and let $\wh{\theta}$ denote the composition $\td{\psi} \circ \td{f}$ of the parameter decoder and the first-stage encoder. The decoder $\td{\psi}$ is the {\em first-stage decoder}. The collection $\{C^n_\theta : \theta \in \Lambda\}$ defines the {\em second-stage codes}. The encoder $\map{f_*}{\cX^n \times \cX^{m_n}}{\td{\cS} \times \cS}$ and the decoder $\map{\varphi_*}{\td{\cS} \times \cS}{\hcX^n}$ of $C^{n,m_n}_*$ are defined as
$$
f_*(X^n_1,X^0_{-m_n+1}) \deq \td{f}(X^0_{-m_n+1})f_{\wh{\theta}(X^0_{-m_n+1})}(X^n_1)
$$
and
$$
\varphi_*(\td{s}s) \deq \varphi_{\td{\psi}(\td{s})}(s), \qquad s \in \cS, \td{s} \in \td{\cS}
$$
respectively. To assess the performance of $C^{n,m_n}_*$, consider the function
\begin{eqnarray*}
&& g(X^n_1,X^0_{-m_n+1}) \deq \rho_n \Big(X^n_1,C^n_{\wh{\theta}(X^0_{-m_n+1})}(X^n_1)\Big) \\
&& \qquad + \lambda \Big[\ell_n \Big(f_{\wh{\theta}(X^0_{-m_n+1})}(X^n_1)\Big)  + \ell_n \Big( \td{f}(X^0_{-m_n+1}) \Big) \Big].
\end{eqnarray*}
The expectation $\E_{\theta_0}\left\{ g(X^n_1,X^0_{-m_n+1}) \right\}$ of $g$ with respect to $P_{\theta_0}$ is precisely the Lagrangian performance of $C^{n,m_n}_*$, at Lagrange multiplier $\lambda$, on the source $P_{\theta_0}$. We consider separately the contributions of the first-stage and the second-stage codes. Define another function $\map{h}{\cX^n \times \cX^{m_n}}{\R^+}$ by
\begin{eqnarray*}
&& h(X^n_1,X^0_{-m_n+1}) \deq \rho_n\Big(X^n_1,C^n_{\wh{\theta}(X^0_{-m_n+1})}(X^n_1)\Big) \\
&& \qquad + \lambda \ell_n\Big(f_{\wh{\theta}(X^0_{-m_n+1})}(X^n_1)\Big),
\end{eqnarray*}
so that $\E_{\theta_0}\Big\{ h(X^n_1,X^0_{-m_n+1}) \Big| X^0_{-m_n+1} \Big\}$ is the (random) Lagrangian performance of the code $C^n_{\wh{\theta}(X^0_{-m_n+1})}$ on $P_{\theta_0}$. Hence,
$$
g(X^n_1,X^0_{-m_n+1}) = h(X^n_1,X^0_{-m_n+1}) + \lambda \ell_n\Big(\td{f}(X^0_{-m+1})\Big),
$$
so, taking expectations, we get
\begin{eqnarray}
&& L_{\theta_0}(C^{n,m_n}_*,\lambda) = \E_{\theta_0}\left\{ h(X^n_1,X^0_{-m_n+1}) \right\} \nonumber\\
&& \qquad + \lambda \E_{\theta_0}\left\{ \ell_n\Big( \td{f}(X^0_{-m_n+1}) \Big) \right\}.
\label{eq:overall_lagrange}
\end{eqnarray}
Our goal is to show that the first term in Eq.~(\ref{eq:overall_lagrange}) converges to the $n$th-order optimum $\wh{L}^n_{\theta_0}(\lambda)$, and that the second term is $o(1)$.

The proof itself is organized as follows. First we motivate the choice of the memory lengths $m_n$ in Section~\ref{ssec:memory}. Then we indicate how to select the database $\bd{C}$ (Section~\ref{ssec:database}) and how to implement the parameter estimator $\td{\theta}$ (Section~\ref{ssec:param_estimation}) and the parameter encoder/decoder pair $(\td{g},\td{\psi})$ (Section~\ref{ssec:1st_stage}). The proof is concluded by estimating the Lagrangian performance of the resulting code (Section~\ref{ssec:performance}) and the fidelity of the source identification at the decoder (Section~\ref{ssec:src_ident}). In the following, (in)equalities involving the relevant random variables are assumed to hold for all realizations and not just a.s., unless specified otherwise.

\begin{figure*}[htb]
\centerline{
\includegraphics[width=0.7\textwidth]{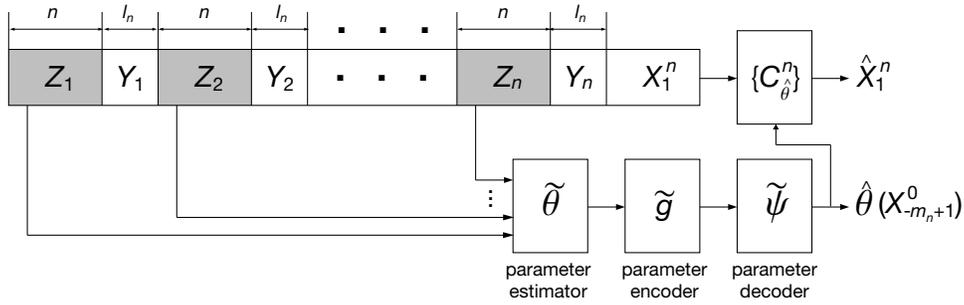}}
\caption{The structure of the code $C^{n,m_n}_*$. The shaded blocks are those used for estimating the source parameters.}
\label{fig:code_structure}
\end{figure*}

\subsection{The memory length}
\label{ssec:memory}

Let $l_n = \lceil n^{(2+\eta)/r} \rceil$, where $r$ is the common decay exponent of the $\beta$-mixing coefficients $\beta_\theta(k)$ in Condition~1, and let $m_n = n(n+l_n)$. We divide the $m_n$-block $X^0_{-m_n+1}$ into $n$ blocks $Z_1,\ldots,Z_n$ of length $n$ interleaved by $n$ blocks $Y_1,\ldots,Y_n$ of length $l_n$ (see Figure~\ref{fig:code_structure}). The parameter estimator $\td{\theta}$, although defined as acting on the entire $X^0_{-m_n+1}$, effectively will make use only of $Z^n = (Z_1,\ldots,Z_n)$. The $Z_j$'s are each distributed according to $P^n_{\theta_0}$, but they are not independent. Thus, the set
$$
\cM = \bigcup^n_{j=1} \{ (j-1)(n+l_n) + 1 \le i \le j(n+l_n) - l_n \}
$$
is the effective memory of $C^{n,m_n}_*$, and the effective memory length is $n^2$.

Let $Q^{(n)}$ denote the marginal distribution of $Z^n$, and let $\td{Q}^{(n)}$ denote the product of $n$ copies of $P^n_{\theta_0}$. We now show that we can approximate $Q^{(n)}$ by $\td{Q}^{(n)}$ in variational distance, increasingly finely with $n$.  Note that both $Q^{(n)}$ and $\td{Q}^{(n)}$ are defined on the $\sigma$-algebra $\cF^{(n)}$, generated by all $X_i$ except those in $Y_1,\ldots,Y_n$, so that $\dvar(Q^{(n)},\td{Q}^{(n)}) = \dvar(Q^{(n)},\td{Q}^{(n)}; \cF^{(n)})$. Therefore, using induction and the definition of the $\beta$-mixing coefficient (cf.~Section~\ref{ssec:sources}), we have
$$
\dvar (Q^{(n)},\td{Q}^{(n)}) \le (n-1) \beta_{\theta_0}(l_n) = O(1/n^{1+\eta}),
$$
where the last equality follows from Condition 1 and from our choice of $l_n$. This in turn implies the following useful fact (see also Lemma~4.1 of Yu \cite{Yu94}), which we shall heavily use in the proof: for any measurable function $\map{\sigma}{\cX^{n^2}}{[0,M]}$ with $M < \infty$,
\begin{eqnarray}
&& \left| \E_{Q^{(n)}}\big\{\sigma(Z^n)\big\} - \E_{\td{Q}^{(n)}} \big\{ \sigma(Z^n)\big\} \right| \le M(n-1)\beta_{\theta_0}(l_n) \nonumber\\
&& \qquad =  O(1/n^{1+\eta}),
\label{eq:blocking_2}
\end{eqnarray}
where the constant hidden in the $O(\cdot)$ notation depends on $M$ and on $\theta_0$.

\subsection{Construction of the database}
\label{ssec:database}

The database, or the first-stage codebook, $\bd{c}$ is constructed by random selection. Let $W$ be a probability distribution on $\Lambda$ which is absolutely continuous with respect to the Lebesgue measure and has an everywhere positive and continuous density $w(\theta)$.  Let $\bd{C} = \{\Theta(i)\}_{i \in \N}$ be a collection of independent random vectors taking values in $\Lambda$, each generated according to $W$ independently of $\bd{X}$. We use $\W$ to  denote the process distribution of $\bd{C}$.

Note that the first-stage codebook is countably infinite, which means that, in principle, both the encoder and the decoder must have unbounded memory in order to store it. This difficulty can be circumvented by using synchronized random number generators at the encoder and at the decoder, so that the entries of $\bd{C}$ can be generated as needed. Thus, by construction, the encoder will generate $T_n$ samples (where $T_n$ is the waiting time) and then communicate (a binary encoding of) $T_n$ to the decoder. Since the decoder's random number generator is synchronized with that of the encoder's, the decoder will be able to recover the required entry of $\bd{C}$.

\subsection{Parameter estimation}
\label{ssec:param_estimation}

The parameter estimator $\map{\td{\theta}}{\cX^{m_n}}{\Lambda}$ is constructed as follows. Because the source $\bd{X}$ is stationary, it suffices to describe the action of $\td{\theta}$ on $X^0_{-m_n+1}$. In the notation of Section~\ref{ssec:main}, let $P_{Z^n}$ be the empirical distribution of $Z^n = (Z_1,\ldots,Z_n)$. For every $\theta \in \Lambda$, define
\begin{eqnarray*}
U_\theta(Z^n) &\deq& \sup_{A \in \cA_n} |P^n_\theta(A) - P_{Z^n}(A)| \\
& \equiv& \sup_{A \in \cA_n} \left| \int_A p_\theta(x^n)d\mu(x^n) - P_{Z^n}(A)\right|,
\end{eqnarray*}
where $\cA_n$ is the Yatracos class defined by the $n$th-order densities $\{p^n_\theta : \theta \in \Lambda\}$ (see Section~\ref{sec:results}). Finally, define $\tilde{\theta}(X^0_{-m+1})$ as any $\theta^* \in \Lambda$ satisfying
$$
U_{\theta^*}(Z^n) < \inf_{\theta \in \Lambda} U_\theta(Z^n) + \frac{1}{n},
$$
where the extra $1/n$ term is there to ensure that at least one such $\theta^*$ exists. This is the so-called {\em minimum-distance (MD) density estimator} of Devroye and Lugosi \cite{DevLug96,DevLug97} (see also Devroye and Gy\"orfi \cite{DevGyo01}), adapted to the dependent-process setting of the present paper. The key property of the MD estimator is that
\begin{equation}
d_n\Big(\td{\theta}(X^0_{m_n+1}),\theta_0\Big) \le 4U_{\theta_0}(Z^n) + \frac{3}{n}
\label{eq:mde_property}
\end{equation}
(see, e.g., Theorem~5.1 of Devroye and Gy\"orfi \cite{DevGyo01}). This holds regardless of whether the samples $Z_1,\ldots,Z_n$ are independent or not.

\subsection{Encoding and decoding of parameter estimates}
\label{ssec:1st_stage}

Next we construct the parameter encoder-decoder pair $(\td{g},\td{\psi})$. Given a $\theta \in \Lambda$, define the {\em waiting time}
$$
T_n(\theta) \deq \inf\{ i \ge 1 : d_n(\theta,\Theta(i)) \le \sqrt{n}\delta_n \},
$$
with the standard convention that the infimum of the empty set is equal to $+\infty$. That is, given a $\theta \in \Lambda$, the parameter encoder looks through the codebook $\bd{C}$ and finds the position of the first $\Theta(i)$ such that the variational distance between the $n$th-order distributions $P^n_{\theta}$ and $P^n_{\Theta(i)}$ is at most $\sqrt{n} \delta_n$. If no such $\Theta(i)$ is found, the encoder sets $T_n = + \infty$. We then define the maps $\td{g}$ and $\td{\psi}$ by
$$
\td{g}(\theta) = \left\{
\begin{array}{ll}
0s({T_n}), & \mbox{ if } T_n < \infty\\
1, & \mbox{ if } T_n = \infty \end{array}\right.
$$
and
$$
\td{\psi}(0s(i)) = \Theta(i), \qquad \td{\psi}(1) = \theta(1)
$$
respectively. Thus, $\td{\cS} = \{0s(i)\} \cup \{1\}$, and the bound
\begin{equation}
\ell(\td{g}(\theta)) \le \log T_n + O(\log \log T_n) 
\label{eq:1st_stage_length}
\end{equation}
holds for every $\theta \in \Lambda$, regardless of whether $T_n$ is finite or infinite.

\subsection{Performance of the code}
\label{ssec:performance}

Given the random codebook $\bd{C}$, the expected Lagrangian performance of our code on the source $P_{\theta_0}$, is
\begin{eqnarray}
&& L_{\theta_0}(C^{n,m_n}_*,\lambda) = \E_{\theta_0}\left\{g\left(X^n_1,X^0_{-m_n+1)}\right)\right\} \nonumber \\
&& \qquad = \E_{\theta_0}\left\{h\left(X^n_1,X^0_{-m_n+1}\right)\right\} \nonumber \\
&& \qquad \qquad + \lambda \E_{\theta_0}\left\{\ell_n\left(\td{f}\left(X^0_{-m_n+1}\right)\right)\right\}.
\label{eq:overall_lagrange_random}
\end{eqnarray}
We now upper-bound the two terms in (\ref{eq:overall_lagrange_random}). We start with the second term.

We need to bound the expectation of the waiting time $T_n = T_n(\td{\theta}(X^0_{-m_n+1}))$. Our strategy borrows some elements from the paper of Kontoyiannis and Zhang \cite{KonZha02}. Consider the probability
$$
q_n \deq W\left(d_n\big(\Theta,\td{\theta}(X^0_{-m_n+1})\big) \le \sqrt{n}\delta_n \right),
$$
which is a random function of $X^0_{-m_n+1}$. From Condition 2, it follows for $n$ sufficiently large that
$$
q_n \ge W\left(\| \Theta - \td{\theta}(X^0_{-m_n+1}) \| \le \delta_n/c_{\td{\theta}}\right),
$$
where $\td{\theta} \equiv \td{\theta}(X^0_{-m_n+1})$. Because the density $w(\theta)$ is everywhere positive, the latter probability is strictly positive for almost all $X^0_{-m_n+1}$, and so $q_n > 0$ eventually almost surely. Thus, the waiting times $T_n$ will be finite eventually almost surely (with respect to both the source $\bd{X}$ and the first-stage codebook $\bd{C}$). Now, if $q_n > 0$, then, conditioned on $X^0_{-m_n+1} = x^0_{-m_n+1}$, the waiting time $T_n$ is a geometric random variable with parameter $q_n$, and it is not hard to show (see, e.g., Lemma~3 of Kontoyiannis and Zhang \cite{KonZha02}) that for any $\epsilon > 0$
$$
\Pr\Big( \log [(T_n - 1)q_n] \ge \epsilon \Big| X^0_{-m_n+1} = x^0_{-m_n + 1} \Big) \le e^{-2^\epsilon}.
$$
Setting $\epsilon = \log(2 \log n)$, we have, for almost all $X^n_{-m_n+1}$, that
\begin{eqnarray*}
&& \Pr\Big( \log [(T_n - 1)q_n] \ge \log(2 \log n) \Big| X^0_{-m_n+1} = x^0_{-m_n+1} \Big) \\
&& \qquad \qquad \le e^{-2\log n} \le n^{-2}.
\end{eqnarray*}
Then, by the Borel--Cantelli lemma,
$$
\log (T_n q_n) \le \log\log n + 2
$$
eventually almost surely, so that
\begin{equation} \label{eq:expected_waittime}
\E_{\theta_0}\left\{\log T_n\right\} \le \log\log n +2 - \E_{\theta_0}\left\{\log q_n\right\}
\end{equation}
for almost every realization of the random codebook $\bd{C}$ and for sufficiently large $n$. We now obtain an asymptotic lower bound on $\E_{\theta_0} \{\log q_n\}$. Define the events
\begin{eqnarray*}
F_n &\deq& \left\{ (d_n\Big(\td{\theta}(X^0_{-m_n+1}),\theta_0\Big) \le \sqrt{n}\delta_n/2 \right\}, \\
G_n &\deq& \left\{ d_n(\Theta,\theta_0) \le \sqrt{n}\delta_n/2 \right\}, \\
H_n &\deq& \left\{ \|\Theta - \theta_0\| \le \delta_n/2c_{\theta_0} \right\}.
\end{eqnarray*}
Then by the triangle inequality we have
$$
F_n \mbox{ and } G_n \quad \Longrightarrow \quad d_n\Big(\Theta,\td{\theta}(X^0_{-m_n+1})\Big) \le \sqrt{n}\delta_n,
$$
and, for $n$ sufficiently large, we can write
\begin{eqnarray*}
q_n &\stackrel{{\rm (a)}}{\ge}& W(G_n) P_{\theta_0}(F_n) \\
&\stackrel{{\rm (b)}}=& W(G_n) Q^{(n)}(F_n) \\
&\stackrel{{\rm (c)}}{\ge}& W(H_n) Q^{(n)}(F_n),
\end{eqnarray*}
where (a) follows from the independence of $\bd{X}$ and $\bd{C}$, (b) follows from the fact that the parameter estimator $\td{\theta}(X^0_{-m_n+1})$ depends only on $Z^n$, and (c) follows from Condition 2 and the fact that $\delta_n \to 0$. Since the density $w$ is everywhere positive and continuous at $\theta_0$, $w(\theta) \ge w(\theta_0)/2$ for all $\theta \in H_n$ for $n$ sufficiently large, so
\begin{equation}
W(H_n) = \int_{H_n} w(\theta)d\theta \ge \frac{1}{2} w(\theta_0) v_k \left(\frac{\delta_n}{2c_{\theta_0}}\right)^k,
\label{eq:prob_bound_1}
\end{equation}
where $v_k$ is the volume of the unit sphere in $\R^k$. Next, the fact that the minimum-distance estimate $\td{\theta}(X^0_{-m_n+1})$ depends only on $Z^n$ implies that the event $F_n$ belongs to the $\sigma$-algebra $\cF^{(n)}$, and from (\ref{eq:blocking_2}) we get
\begin{equation}
Q^{(n)}(F_n) \ge \td{Q}^{(n)}(F_n) - O(1/n^{1+\eta}).
\label{eq:mde_bound_1}
\end{equation}
Under $\td{Q}^{(n)}$, the $n$-blocks $Z_1,\ldots,Z_n$ are i.i.d. according to $P^n_{\theta_0}$, and we can invoke the Vapnik--Chervonenkis machinery to lower-bound $\td{Q}^{(n)}(F_n)$. In the notation of Sec.~\ref{ssec:param_estimation}, define the event
$$
I_n \deq \left\{  4U_{\theta_0}(Z^n) + \frac{3}{n} \le \frac{\sqrt{n}\delta_n}{2} \right\}.
$$
Then $I_n$ implies $F_n$ by (\ref{eq:mde_property}), and
\begin{equation}
\td{Q}^{(n)}(F^c_n) \le \td{Q}^{(n)}(I^c_n) \le 8n^{\sV(\cA_n)} e^{-n(\sqrt{n}\delta_n - 6/n)^2/2048},
\label{eq:mde_bound_2}
\end{equation}
where the second bound is by the Vapnik--Chervonenkis inequality (\ref{eq:vcbound_1}) of Lemma~\ref{lm:vc}. Combining the bounds (\ref{eq:mde_bound_1}) and (\ref{eq:mde_bound_2}) and using Condition 1, we obtain
\begin{equation}
P^m_\theta(F_n) \ge 1 - 8n^{\sV(\cA_n)}e^{-n(\sqrt{n}\delta_n - 6/n)^2/2048} - O(1/n^{1+\eta})
\label{eq:prob_bound_2}
\end{equation}
Now, if we choose
$$
\delta_n = \frac{\sqrt{2048 (\sV(\cA_n) + 1)\ln n}}{n} + \frac{6}{n^{3/2}},
$$
then the right-hand side of (\ref{eq:prob_bound_2}) can be further lower-bounded by $1-O(1/n)$. Combining this with (\ref{eq:prob_bound_1}), taking logarithms, and then taking expectations, we obtain
\begin{eqnarray*}
\lefteqn{\E_{\theta_0}\{\log q_n\}} \\
 &\ge& \log (1 - O(1/n)) + k \log \delta_n + 2^{c(k,\theta_0)} \\
 &=& \log(1-O(1/n))  \\
 && \quad + k \log \left[\sqrt{2048(\sV(\cA_n) + 1)n\ln n} + 6\right] \\
 && \quad  + \frac{3k}{2}\log \frac{1}{n} + c(k,\theta_0) \\
&\ge& \log(1- O(1/n)) + \frac{3k}{2}\log \frac{1}{n} + c(k,\theta_0),
\end{eqnarray*}
where $c(k,\theta_0)$ is a constant that depends only on $k$ and $\theta_0$. Using this and (\ref{eq:expected_waittime}), we get that
$$
\E_{\theta_0}\{\log T_n\} \le \log\log n + O(\log n)
$$
for $\W$-almost every realization of the random codebook $\bd{C}$, for $n$ sufficiently large. Together with (\ref{eq:1st_stage_length}), this implies that
\begin{eqnarray*}
\lefteqn{\E_{\theta_0}\left\{ \ell_n\left(\td{f}\left(X^0_{-m_n+1}\right)\right)\right\}}\\
& =& O\left(\frac{\log n}{n} \right)  + O\left(\frac{\log \log n}{n}\right) +\frac{3}{n} + o(1)
\end{eqnarray*}
for $\W$-almost all realizations of the first-stage codebook.

We now turn to the first term in (\ref{eq:overall_lagrange_random}). Recall that, for each $\theta \in \Lambda$, the code $C^n_\theta$ is $n$th-order optimal for $P_\theta$. Using this fact together with the boundedness of the distortion measure $\rho$, we can invoke Lemma~\ref{lm:finite_codebook} in the Appendix and assume without loss of generality that each $C^n_\theta$ has a finite codebook (of size not exceeding $2^{n\rho_{\max}/\lambda}$), and each codevector can be described by a binary string of no more than $2n\rho_{\max}/\lambda$ bits. Hence, $h(X^n_1,X^0_{-m_n+1}) \le 3\rho_{\max}$. Let $P^-$ and $P^+$ be the marginal distributions of $P_{\theta_0}$ on $\sigma(X^0_{-\infty})$ and $\sigma(X^\infty_1)$, respectively. Note that $h(X^n,X^0_{-m_n+1})$ does not depend on $X^0_{-l_n+1}$. This, together with Condition 1 and the choice of $l_n$, implies that
\begin{eqnarray*}
&& \E_{\theta_0}\left\{h\left(X^n_1,X^0_{-m_n+1}\right)\right\} \\
&& \qquad \le \E_{P^- \times P^+} \left\{h\left(X^n_1,X^0_{-m_n+1}\right)\right\} + \beta_{\theta_0}(l_n) \\
&& \qquad = \E_{P^- \times P^+} \left\{h\left(X^n_1,X^0_{-m_n+1}\right)\right\} + O(1/n^{2+\eta}).
\end{eqnarray*}
Furthermore, 
\begin{eqnarray*}
&& \E_{P^- \times P^+} \left\{h\left(X^n_1,X^0_{-m_n+1}\right)\right\} \\
&& \qquad = \int_{\cX^n \times \cX^{m_n}} h(x^n,z^{m_n}) dP_{\theta_0}(x^n) dP_{\theta_0}(z^{m_n}) \\
&& \qquad \stackrel{{\rm (a)}}{=} \int_{\cX^{m_n}} \E_{\theta_0}\left\{h(X^n_1,z^{m_n})\right\} dP_{\theta_0}(z^{m_n}) \\
&& \qquad \stackrel{{\rm (b)}}{=} \E_{\theta_0}\left\{L_{\theta_0}\Big(C^n_{\wh{\theta}(X^0_{-m_n+1})},\lambda\Big)\right\},
\end{eqnarray*}
where (a) follows by Fubini's theorem and the boundedness of $h$, while (b) follows from the definition of $h$. The Lagrangian performance of the code $C^n_{\wh{\theta}}$, where $\wh{\theta} = \wh{\theta}(X^0_{-m_n+1})$, can be further bounded as
\begin{eqnarray*}
&& L_{\theta_0}\left(C^n_{\wh{\theta}},\lambda\right) \\
&& \quad \stackrel{{\rm (a)}}{\le} L_{\wh{\theta}}\left(C^n_{\wh{\theta}},\lambda\right) + 3\rho_{\max}d_n\Big(\wh{\theta}(X^0_{-m_n+1}),\theta_0\Big) \\
&& \quad \stackrel{{\rm (b)}}{=} \wh{L}^n_{\wh{\theta}}(\lambda) + 3\rho_{\max}d_n\Big(\wh{\theta}(X^0_{-m_n+1}),\theta_0\Big)\\
&& \quad \stackrel{{\rm (c)}}{\le} \wh{L}^n_{\theta_0}(\lambda) + 4\rho_{\max}d_n\Big(\wh{\theta}(X^0_{-m_n+1}),\theta_0\Big)\\
&& \quad \stackrel{{\rm (d)}}{\le} \wh{L}^n_{\theta_0}(\lambda) + 4\rho_{\max}\Big[d_n\Big(\wh{\theta}(X^0_{-m_n+1}),\td{\theta}(X^0_{-m_n+1})\Big) \\
&& \qquad \qquad \qquad \qquad + d_n\Big(\td{\theta}(X^0_{-m_n+1}),\theta_0\Big)\Big],
\end{eqnarray*}
where (a) follows from Lemma~\ref{lm:finite_codebook} in the Appendix, (b) follows from the $n$th-order optimality of $C^n_{\wh{\theta}}$ for $P^n_{\wh{\theta}}$, (c) follows, overbounding slightly, from the Lagrangian mismatch bound of Lemma~\ref{lm:mismatch} in the Appendix, and (d) follows from the triangle inequality. Taking expectations, we obtain
\begin{eqnarray}
&& \E_{\theta_0}\left\{ L_{\theta_0}\big(C^n_{\wh{\theta}(X^0_{-m_n+1})},\lambda\big)\right\} \le \wh{L}^n_{\theta_0}(\lambda) \nonumber \\
&& \qquad + 4\rho_{\max}\cdot\E_{\theta_0}\Big\{d_n\Big(\wh{\theta}(X^0_{-m_n+1}),\td{\theta}(X^0_{-m_n+1})\Big) \nonumber \\
&& \qquad + d_n\Big(\td{\theta}(X^0_{-m_n+1}),\theta_0\Big)\Big\}.
\label{eq:lagrange_bound}
\end{eqnarray}
The second $d_n(\cdot,\cdot)$ term in (\ref{eq:lagrange_bound}) can be interpreted as the estimation error due to estimating $P^n_{\theta_0}$ by $P^n_{\td{\theta}}$, while the first $d_n(\cdot,\cdot)$ is the approximation error due to quantization of the parameter estimate $\td{\theta}$. We examine the estimation error first. Using (\ref{eq:mde_property}), we can write
\begin{equation}
\E_{\theta_0}\left\{\dvar\Big(P^n_{\theta^*(X^0_{-m+1})},P^n_{\theta_0}\Big)\right\} \le 4\E_{\theta_0}\{U_{\theta_0}(Z^n)\} + \frac{3}{n}.
\label{eq:mde_error}
\end{equation}
Now, each $Z_j$ is distributed according to $P^n_{\theta_0}$, and we can approximate the expectation of $U_{\theta_0}(Z^n)$ with respect to $Q^{(n)}$ by the expectation of $U_{\theta_0}(Z^n)$ with respect to the product measure $\td{Q}^{(n)}$:
\begin{eqnarray*}
\E_{Q^{(n)}}\left\{U_{\theta_0}(Z^n)\right\} &\le& \E_{\td{Q}^{(n)}}\left\{U_{\theta_0}(Z^n)\right\} + (n-1)\beta_\theta(l_n) \\
&\le& c\sqrt{\frac{\sV(\cA_n)\log n}{n}} + O\left(\frac{1}{n^{1+\eta}}\right) \\
&=& O\left(\sqrt{\frac{\sV(\cA_n)\log n}{n}}\right),
\end{eqnarray*}
where the second estimate follows from the Vapnik--Chervonenkis inequality (\ref{eq:vcbound_2}) and from the choice of $l_n$. This, together with (\ref{eq:mde_error}), yields
\begin{equation}
\E_{\theta_0}\left\{\dvar\Big(P^n_{\theta^*(X^0_{-m+1})},P^n_{\theta_0}\Big)\right\}  = O\left(\sqrt{\frac{\sV(\cA_n)\log n}{n}}\right).
\label{eq:estimation_error}
\end{equation}
As for the first $d_n(\cdot,\cdot)$ term in (\ref{eq:lagrange_bound}), we have, by construction of the first-stage encoder, that
\begin{eqnarray}
&& d_n\Big(\wh{\theta}(X^0_{-m_n+1}),\td{\theta}(X^0_{-m_n+1})\Big) \nonumber \\
&& \qquad \qquad \le \sqrt{n}\delta_n = O\left(\sqrt{\frac{\sV(\cA_n)\log n}{n}}\right)
\label{eq:approximation_error}
\end{eqnarray}
eventually almost surely, so the corresponding expectation is $O \big(\sqrt{\sV(\cA_n)\log n/n} \big)$ as well. Summing the estimates (\ref{eq:estimation_error}) and (\ref{eq:approximation_error}), we obtain
$$
\E_{\theta_0} \left\{ h\big(X^n_1,X^0_{-m_n+1}\big) \right\} = \wh{L}^n_{\theta_0}(\lambda) + O\left(\sqrt{\frac{\sV(\cA_n)\log n}{n}}\right).
$$
Finally, putting everything together, we see that, eventually,
\begin{eqnarray}
&& L_{\theta_0}\Big(C^{n,m_n}_*\Big) = \wh{L}^n_{\theta_0}(\lambda) + O\left(\sqrt{\frac{\sV(\cA_n)\log n}{n}}\right) \nonumber\\
&& \,\, + \lambda\Big[O\Big(\frac{\log n}{n} \Big) + O\Big(\frac{\log \log n}{n}\Big)   + \frac{3}{n} + o(1)\Big]
\end{eqnarray}
for $\W$-almost every realization of the first-stage codebook $\bd{C}$. This proves (\ref{eq:zeromem_universality}), and hence (\ref{eq:universality}).

\subsection{Identification of the active source}
\label{ssec:src_ident}

We have seen that the expected variational distance $\E_{\theta_0}\left\{d_n\big(\theta_0,\wh{\theta}(X^0_{-m_n+1})\big)\right\}$ between the $n$-dimensional marginals of the true source $P_{\theta_0}$ and the estimated source $P_{\wh{\theta}(X^0_{-m_n+1})}$ converges to zero as $\sqrt{\sV(\cA_n)\log n/n}$. We wish to show that this convergence also holds eventually with probability one, i.e.,
\begin{equation}
d_n(\theta_0,\wh{\theta}(X^0_{-m_n+1})) = O\left(\sqrt{\frac{\sV(\cA_n)\log n}{n}}\right)
\label{eq:as_bound}
\end{equation}
$P_{\theta_0}$-almost surely.

Given an $\epsilon > 0$, we have by the triangle inequality that $d_n(\theta_0,\wh{\theta}(X^0_{-m_n+1})) > \epsilon$ implies
$$
d_n \Big(\theta_0,\td{\theta}(X^0_{-m_n+1}) \Big) + d_n \Big (\td{\theta}(X^0_{-m_n+1}),\wh{\theta}(X^0_{-m_n+1}) \Big) > \epsilon,
$$
where $\td{\theta}(X^0_{-m_n+1})$ is the minimum-distance estimate of $P^n_{\theta_0}$ from $X^0_{-m_n+1}$ (cf.~Section~\ref{ssec:1st_stage}). Recalling our construction of the first-stage encoder, we see that this further implies
$$
d_n \Big(\theta_0,\td{\theta}(X^0_{-m_n+1}) \Big) > \epsilon - \sqrt{n}\delta_n.
$$
Finally, using the property (\ref{eq:mde_property}) of the minimum-distance estimator, we obtain that
$$
d_n \Big(\theta_0,\wh{\theta}(X^0_{-m_n+1}) \Big) > \epsilon
$$
implies
$$
U_{\theta_0}(Z^n) > \frac{1}{4}\left(\epsilon - \sqrt{n}\delta_n - \frac{3}{n}\right).
$$
Therefore,
\begin{eqnarray}
\lefteqn{Q^{(n)} \left\{ d_n \Big(\theta_0,\wh{\theta}(X^0_{-m_n+1}) \Big) > \epsilon \right\}} \nonumber \\
&\le& Q^{(n)} \left\{ U_{\theta_0}(Z^n) > \frac{1}{4}\left(\epsilon - \sqrt{n}\delta_n - \frac{3}{n}\right) \right\} \nonumber\\
&\stackrel{{\rm (a)}}{\le}& \td{Q}^{(n)} \left\{ U_{\theta_0}(Z^n) > \frac{1}{4}\left(\epsilon - \sqrt{n}\delta_n - \frac{3}{n}\right) \right\} \nonumber\\
&& \qquad \qquad  + (n-1)\beta_{\theta_0}(l_n) \nonumber\\
&\stackrel{{\rm (b)}}{\le}& 8n^{\sV(\cA_n)} \exp\left(-\frac{n(\epsilon - \sqrt{n}\delta_n - 3/n)^2}{512}\right) \nonumber\\
&& \qquad \qquad + (n-1)\beta_{\theta_0}(l_n),
\label{eq:prob_bound}
\end{eqnarray}
where (a) follows, as before, from the definition of the $\beta$-mixing coefficient and (b) follows by the Vapnik--Chervonenkis inequality. Now, if we choose
$$
\epsilon_n = \sqrt{\frac{512(\sV(\cA_n)\ln n + n \delta)}{n}} + \sqrt{n}\delta_n + \frac{3}{n}
$$
for an arbitrary small $\delta > 0$, then (\ref{eq:prob_bound}) can be further upper-bounded by $8e^{-n\delta} + \sum_n n\beta_\theta(\ell_n)$, which, owing to Condition 1 and the choice $l_n = \lceil n^{(2+\eta)/r} \rceil$, is summable in $n$. Thus,
$$
\sum_n Q^{(n)}\left\{ d_n \Big(\theta_\theta,\wh{\theta}(X^0_{-m_n+1}) \Big) > \epsilon_n \right\} < \infty,
$$
and we obtain (\ref{eq:as_bound}) by the Borel--Cantelli lemma.

\section{Examples}
\label{sec:examples}

\subsection{Stationary memoryless sources}
\label{ssec:iid}

As a basic check, let us see how Theorem~\ref{thm:main} applies to stationary memoryless (i.i.d.) sources. Let $\cX = \R$, and let $\{P_\theta : \theta \in \Lambda\}$ be the collection of all Gaussian i.i.d. processes, where
$$
\Lambda = \{(m,\sigma) : m \in \R, 0 < \sigma < \infty \} \subset \R^2.
$$
Then the $n$-dimensional marginal for a given $\theta = (m,\sigma)$ has the Gaussian density
$$
p_\theta(x^n) = \frac{1}{(2\pi\sigma^2)^{n/2}} \prod^n_{i=1} e^{-(x_i - m)^2/2\sigma^2}
$$
with respect to the Lebesgue measure. This class of sources trivially satisfies Condition 1 with $r = +\infty$, and it remains to check Conditions 2 and 3.

To check Condition 2, let us examine the normalized $n$th-order relative entropy between $P_\theta$ and $P_{\theta'}$, with $\theta = (m,\sigma)$ and $\theta' = (m',\sigma')$. Because the sources are i.i.d.,
\begin{eqnarray*}
&& D_n(P_\theta \| P_{\theta'}) = D(P^1_\theta \| P^1_{\theta'} ) \\
&& \quad = \frac{1}{2}\left(\ln \left(\frac{\sigma}{\sigma'}\right)^2 + \left(\frac{\sigma'}{\sigma}\right)^2 + \frac{(m - m')^2}{{\sigma'}^2} - 1\right). 
\end{eqnarray*}
Applying the inequality $\ln x \le x - 1$ and some straightforward algebra, we get the bound
\begin{eqnarray*}
D_n(P_\theta \| P_{\theta'}) &\le& \left( \frac{\sigma + \sigma'}{\sigma}\right)^2 \frac{(\sigma - \sigma')^2}{2{\sigma'}^2} +  \frac{(m-m')^2}{2{\sigma'}^2} \\
&\le& \left(1 + \frac{\sigma'}{\sigma}\right)^2 \frac{ \| \theta - \theta' \|^2 } {2{\sigma'}^2}.
\end{eqnarray*}
Now fix a small $\delta \in (0,\sigma)$, and suppose that $\| \theta - \theta' \|  < \delta$. Then $|\sigma - \sigma'| < \delta$, so we can further upper-bound $D_n(P_\theta \| P_{\theta'})$ by
$$
D_n(P_\theta \| P_{\theta'}) \le \frac{9}{2(\sigma - \delta)^2} \| \theta - \theta' \|^2.
$$
Thus, for a given $\theta \in \Lambda$, we see that
$$
D_n(P_\theta \| P_{\theta'}) \le \frac{c^2_\theta}{2} \| \theta - \theta'\|^2
$$
for all $\theta'$ in the open ball of radius $\delta$ around $\theta$, with $c_\theta \deq 3/(\sigma - \delta)$. Using Pinsker's inequality, we have
$$
\frac{d_n(\theta,\theta')}{\sqrt{n}} \le \sqrt{2 D_n(P_\theta \| P_{\theta'})} \le c_\theta \| \theta - \theta' ||
$$
for all $n$. Thus, Condition 2 holds.

To check Condition 3, note that, for each $n$, the Yatracos class $\cA_n$ consists of all sets of the form
\begin{eqnarray}
&& \Bigg\{ x^n \in \R^n : \ln \sigma^2 - \ln {\sigma'}^2 + \frac{1}{\sigma^2} \sum^n_{i=1} (x_i - m)^2 \nonumber\\
&& \qquad  \qquad - \frac{1}{{\sigma'}^2}\sum^n_{i=1} (x_i - m')^2 > 0 \Bigg\}
\label{eq:yatracos_gauss}
\end{eqnarray}
for all $ m,m' \in \R; \sigma,\sigma' \in (0,\infty)$. Let $\alpha \deq \ln \sigma^2$ and $\alpha' \deq \ln {\sigma'}^2$. Then we can rewrite (\ref{eq:yatracos_gauss}) as
\begin{eqnarray*}
&& \Bigg\{ x^n \in \R^n : \alpha - \alpha' +\frac{1}{\sigma^2} \sum^n_{i=1}(x_i - m)^2 \\
&& \qquad \qquad - \frac{1}{{\sigma'}^2} \sum^n_{i=1}(x_i - m')^2 > 0 \Bigg\}.
\end{eqnarray*}
This is the set of all $x^n \in \R^n$ such that
$$
\Pi(x^n,\alpha,\alpha',1/\sigma^2, 1/{\sigma'}^2,m,m') > 0,
$$
where $\Pi(x^n,\cdot)$ is a third-degree polynomial in the six parameters $(\alpha,\alpha',1/\sigma^2,1/{\sigma'}^2,m,m')$. It then follows from Lemma~\ref{lm:karpinski_macintyre} that $\cA_n$ is a VC class with $\sV(\cA_n) \le 12\log(12e)$. Therefore, Condition 3 holds as well.

\subsection{Autoregressive sources}
\label{ssec:ar}

Again, let $\cX = \R$ and consider the case when $\bd{X}$ is a Gaussian autoregressive source of order $p$, i.e., it is the output of an autoregressive filter of order $p$ driven by white Gaussian noise. Then there exist $p$ real parameters $a_1,\ldots,a_p$ (the filter coefficients), such that
$$
X_n = - \sum^p_{i=1} a_i X_{n-i} + Y_n,\qquad \forall n \in \N
$$
where $\bd{Y} = \{Y_i\}_{i \in \Z}$ is an i.i.d. Gaussian process with zero mean and unit variance. Let $\Lambda \subset \R^p$ be the set of all $a_1,\ldots,a_p$, such that the roots of the polynomial $A(z) = \sum^p_{i=0}a_i z^i$, where $a_0 \equiv 1$, lie outside the unit circle in the complex plane. This ensures that $\bd{X}$ is a stationary process. We now proceed to check that Conditions 1--3 of Section~\ref{sec:results} are satisfied.

The distribution of each $Y_i$ is absolutely continuous, and we can invoke the result of Mokkadem \cite{Mok88} to conclude that, for each $\theta \in \Lambda$, the process $\bd{X}$ is {\em geometrically mixing}, i.e., for every $\theta \in \Lambda$, there exists some $\gamma = \gamma(\theta) \in (0,1)$, such that $\beta_\theta(k) = O(\gamma^k)$. Now, for any fixed $r > 0$, $\gamma^k \le k^{-r}$ for $k$ sufficiently large, so Condition 1 holds. 

To check Condition 2, note that, for each $\theta \in \Lambda$, the Fisher information matrix $I_n(\theta)$ is independent of $n$ (see, e.g., Section~6 of Klein and Spreij \cite{KleSpr06}). Thus, the asymptotic Fisher information matrix $I(\theta) = \lim_{n \to \infty} I_n(\theta)$ exists and is nonsingular \cite[Theorem~6.1]{KleSpr06}, so, recalling the discussion in Section~\ref{sec:results}, we conclude that Condition 2 holds also.

To verify Condition 3, consider the $n$-dimensional marginal $P_\theta(x^n)$, which has the Gaussian density
$$
p_\theta(x^n) = \frac{1}{(2\pi\det R_n(\theta))^{n/2}} e^{-\frac{1}{2}{x^n}^T R^{-1}_n(\theta) x^n},
$$
where $R_n(\theta) \equiv \E_\theta\big\{ (X^n_1)^T X^n_1 \big\}$ is the $n$th-order autocorrelation matrix of $\bd{X}$. Thus, the Yatracos class $\cA_n$ consists of sets of the form
\begin{eqnarray*}
&& A_{\theta,\theta'} = \Bigg\{x^n \in \R^n : \frac{n}{2} \ln \det R^{-1}_n(\theta) - \frac{1}{2} {x^n}^T R^{-1}_n(\theta) x^n \\
&& \qquad \qquad > \frac{n}{2} \ln \det R^{-1}_n(\theta') - \frac{1}{2} {x^n}^T R^{-1}_n(\theta') x^n \Bigg\}
\end{eqnarray*}
for all $\theta,\theta' \in \Lambda$. Now, for every $\theta \in \Lambda$, let $\bar{\theta} \deq (\theta,\ln \det R^{-1}_n(\theta))$. Since $\ln \det R^{-1}_n(\theta)$ is uniquely determined by $\theta$, we have $A_{\theta,\theta'} = A_{\bar{\theta},\bar{\theta}'}$ for all $\theta, \theta' \in \Lambda$. Using this fact, as well as the easily established fact that the entries of the inverse covariance matrix $R^{-1}_n(\theta)$ are second-degree polynomials in the filter coefficients $a_1,\ldots,a_p$, we see that, for each $x^n$, the condition $x^n \in A_{\theta,\theta'}$ can be expressed as $\Pi(x^n,\bar{\theta}) > 0$, where $\Pi(x^n,\cdot)$ is quadratic in the $2p+2$ real variables $\bar{\theta}_1,\ldots,\bar{\theta}_{p+1},\bar{\theta}'_1,\ldots,\bar{\theta}'_{p+1}$. Thus, we can apply Lemma~\ref{lm:karpinski_macintyre} to conclude that $\sV(\cA_n) \le (4p+4)\log(8e)$. Therefore, Condition 3 is satisfied as well.

\subsection{Hidden Markov processes}
\label{ssec:hmp}

A hidden Markov process (or a hidden Markov model, see, e.g., \cite{BicRitRyd98}) is a discrete-time bivariate random process $\{(S_i,X_i)\}$, where $\bd{S} = \{S_i\}$ is a homogeneous Markov chain and $\bd{X} = \{X_i\}$ is a sequence of random variables which are conditionally independent given $\bd{S}$, and the conditional distribution of $X_n$ is time-invariant and depends on $\bd{S}$ only through $S_n$. The Markov chain $\bd{S}$, also called the {\em regime}, is not available for observation. The observable component $\bd{X}$ is the source of interest. In information theory (see, e.g., \cite{EphMer02} and references therein), a hidden Markov process is a discrete-time finite-state homogeneous Markov chain $\bd{S}$, observed through a discrete-time memoryless channel, so that $\bd{X} = \{X_i\}$ is the observation sequence at the output of the channel.

Let $M$ denote the number of states of $\bd{S}$. We assume without loss of generality that the state space $\cS$ of $\bd{S}$ is the set $\{1,2,\ldots,M\}$. Let $A = [a_{ij}]_{i,j = 1,\ldots,M}$ denote the $M \times M$ transition matrix of $\bd{S}$, where $a_{ij} \deq \Pr(S_{t+1} = j | S_t = i)$. If $A$ is ergodic (i.e., irreducible and aperiodic), then there exists a unique probability distribution $\pi$ on $\cS$ such that $\pi = \pi A$ (the {\em stationary distribution} of $\bd{S}$), see, e.g., Section~8 of Billingsley \cite{Bil95}. Because in this paper we deal with two-sided random processes, we assume that $\bd{S}$ has been initialized with its stationary distribution at some time sufficiently far away in the past, and can therefore be thought of as a two-sided stationary process. Now consider a discrete-time memoryless channel with input alphabet $\cS$ and output (observation) alphabet $\cX = \R^d$ for some $d < \infty$. It is specified by a set $\{ p(\cdot|s) : s = 1,2,\ldots,M \}$ of transition densities (with respect to $\mu$, the Lebesgue measure on $\R^d$). The channel output sequence $\bd{X}$ is the source of interest.

Let us take as the parameter space $\Lambda \subset \R^{M \times M}$ the set of all $M\times M$ transition matrices $[a_{ij}]$, such that all $a_{ij} > a_0$ for some fixed $a_0 > 0$. For each $\theta = [a_{ij}] \in \Lambda$ and each $n \in \N$, the density $dP^n_\theta/d\mu^n$ is given by
$$
p_\theta(x^n) = \sum_{s^n \in \cS^n} \prod^n_{i=1} a_{s_{i-1}s_i} p(x_i|s_i),
$$
where $a_{s_0s} \equiv \pi_s$ for every $s \in \cS$. We assume that the channel transition densities $p(\cdot|s), s \in \cS$, are fixed {\em a priori}, and do not include them in the parametric description of the sources. We do require, though, that
$$
\sum_{s \in \cS} p(x|s) > 0, \qquad \forall x \in \cX
$$
and
$$
\E_\theta\left\{ \log \sum_{s \in \cS} p(X|s) \right\} < \infty, \qquad \forall \theta \in \Lambda.
$$
We now proceed to verify that Conditions 1--3 of Section~\ref{sec:results} are met.

Let $p^{(n)}_{ij}  = \Pr(S_{t+n} = j | S_t = i)$ denote the $n$-step transition probability for states $i,j \in \cS$. The positivity of $A$ implies that the Markov chain $\bd{S}$ is {\em geometrically ergodic}, i.e.,
\begin{equation}
|p^{(n)}_{ij} - \pi_j| \le C\gamma^n, \qquad \forall i,j \in \cS; \forall n \in \N
\label{eq:geometric_ergodicity}
\end{equation}
where $C \ge 0$ and $0 \le \gamma < 1$, see Theorem~8.9 of Billingsley \cite{Bil95}. Note that (\ref{eq:geometric_ergodicity}) implies that
$$
\dvar(p^{(n)}(\cdot|i),\pi) \le MC\gamma^n.
$$
This in turn implies that the sequence $\bd{S} = \{S_i\}$ is exponentially $\beta$-mixing, see Theorem~3.10 of Vidyasagar \cite{Vid03}. Now, one can show (see Section~3.5.3 of Vidyasagar \cite{Vid03}) that there exists a measurable mapping $\map{F}{\cS \times [0,1]}{\cX}$, such that
$X_i = F(S_i,U_i)$, where $\bd{U} = \{ U_i \}$ is an i.i.d. sequence of random variables distributed uniformly on $[0,1]$, independently of $\bd{S}$. It is not hard to show that, if $\bd{S}$ is exponentially $\beta$-mixing, then so is the bivariate process $\{(S_i,U_i)\}$. Finally, because $X_i$ is given by a time-invariant deterministic function of $(S_i,U_i)$, the $\beta$-mixing coefficients of $\bd{X}$ are bounded by the corresponding $\beta$-mixing coefficients of $(\bd{X},\bd{U})$, and so $\bd{X}$ is exponentially $\beta$-mixing as well. Thus, for each $\theta \in \Lambda$, there exists a $\gamma = \gamma(\theta) \in [0,1)$, such that $\beta_\theta(k) = O(\gamma^k)$, and consequently Condition~1 holds.

To show that Condition 2 holds, we again examine the asymptotic behavior of the Fisher information matrix $I_n(\theta)$ as $n \to \infty$. Under our assumptions on the state transition matrices in $\Lambda$ and on the channel transition densities $\{p(\cdot|s) : s \in \cS\}$, we can invoke the results of Section~6.2 in Douc, Moulines and Ryd\'en \cite{DouMouRyd04} to conclude that the asymptotic Fisher information matrix $I(\theta) = \lim_{n \to \infty} I_n(\theta)$ exists (though it is not necessarily nonsingular). Thus, Condition~2 is satisfied.

Finally we check Condition 3. The Yatracos class $\cA_n$ consists of all sets of the form
\begin{eqnarray*}
&& A_{\theta,\theta'} = \Bigg\{ x^n \in \cX^n : \sum_{s^n \in \cS^n} \left( \prod^n_{i=1} a_{s_{i-1}s_i} - \prod^n_{i=1} a'_{s_{i-1}s_i} \right) \\
&& \qquad \qquad \qquad \qquad \times \prod^n_{j=1} p(x_j|s_j) > 0 \Bigg\}
\end{eqnarray*}
for all $\theta = [a_{ij}], \theta' = [a'_{ij}] \in \Lambda$. The condition $x^n \in A_{\theta,\theta'}$ can be written as $\Pi(x^n,\theta,\theta') > 0$, where for each $x^n$, $\Pi(x^n,\theta,\theta')$ is a polynomial of degree $n$ in the $2M^2$ parameters $a_{ij},a'_{kl}$, $1 \le i,j,k,l \le M$. Thus, Lemma~\ref{lm:karpinski_macintyre} implies that $\sV(\cA_n) \le 4M^2\log(4en)$, so Condition~3 is satisfied as well.

\section{Conclusions and future directions}
\label{sec:conclusion}

We have shown that, given a parametric family of stationary mixing sources satisfying some regularity conditions, there exists a universal scheme for joint lossy compression and source identification, with the $n$th-order Lagrangian redundancy and the variational distance between $n$-dimensional marginals of the true and the estimated source both converging to zero as $\sqrt{V_n\log n/n}$, as the block length $n$ tends to infinity. The sequence $\{V_n\}$ quantifies the learnability of the $n$-dimensional marginals. This generalizes our previous results from \cite{Rag05,Rag06} for i.i.d. sources.

We can outline some directions for future research.
\begin{itemize}
\item Both in our earlier work \cite{Rag05,Rag06} and in the present paper, we assume that the dimension of the parameter space is known {\em a priori}. It would be of interest to consider the case when the parameter space is finite-dimensional, but its dimension is not known. Thus, we would have a hierarchical model class $\bigcup^\infty_{k=1}\{P_\theta : \theta \in \Lambda^{(k)} \}$, where, for each $k$, $\Lambda^{(k)}$ is an open subset of $\R^k$, and we could use a complexity regularization technique, such as ``structural risk minimization" (see, e.g., Lugosi and Zeger \cite{LugZeg96} or Chapter 6 of Vapnik \cite{Vap98}), to adaptively trade off the estimation and the approximation errors.
\item The minimum-distance density estimator of Devroye and Lugosi \cite{DevLug96,DevLug97}, which plays the key role in our scheme both here and in \cite{Rag05,Rag06}, is not easy to implement in practice, especially for multidimensional alphabets. On the other hand, there are two-stage universal schemes, such as that of Chou, Effros and Gray \cite{ChoEffGra96}, which do not require memory and select the second-stage code based on pointwise, rather than average, behavior of the source. These schemes, however, are geared toward compression, and do not emphasize identification. It would be worthwhile to devise practically implementable universal schemes that strike a reasonable compromise between these two objectives.
\item Finally, neither here nor in our earlier work \cite{Rag05,Rag06} have we considered the issues of optimality. It would be of interest to obtain lower bounds on the performance of any universal scheme for joint lossy compression and identification, say, in the spirit of minimax lower bounds in statistical learning theory (cf., e.g., Chapter 14 of Devroye, Gy\"orfi and Lugosi \cite{DevGyoLug96}).
\end{itemize}

Conceptually, our results indicate that links between statistical modeling (parameter estimation) and universal source coding, exploited in the lossless case by Rissanen \cite{Ris84,Ris96}, are present in the domain of lossy coding as well. We should also mention that another modeling-based approach to universal lossy source coding, due to Kontoyiannis and others (see, e.g., Madiman and Kontoyiannis \cite{MadKon04} and references therein), treats code selection as a statistical estimation problem over a class of model distributions in the {\em reproduction space}. This approach, while closer in spirit to Rissanen's Minimum Description Length (MDL) principle \cite{BarRisYu98}, does not address the problem of joint source coding and identification, but it provides a complementary perspective on the connections between lossy source coding and statistical modeling.

\appendix
\renewcommand{\theequation}{A.\arabic{equation}}
\setcounter{equation}{0}
\renewcommand{\thelemma}{A.\arabic{lemma}}

\section{Properties of Lagrange-optimal variable-rate quantizers}

In this Appendix, we detail some properties of Lagrange-optimal variable-rate vector quantizers. Our exposition is patterned on the work of Linder \cite{Lin01}, with appropriate modifications.

As elsewhere in the paper, let $\cX$ be the source alphabet and $\hcX$ the reproduction alphabet, both assumed to be Polish spaces. As before, let the distortion function $\rho$ be induced by a $\rho_{\max}$-bounded metric on a Polish metric space $\cY$ containing $\cX \cup \hcX$. For every $n = 1,2,\ldots$, define the metric $\rho_n$ on $\cY^n$ by
$$
\rho_n(y^n,u^n) \deq \frac{1}{n}\sum^n_{i=1} \rho(y_i,u_i).
$$
For any pair $P^{(1)},P^{(2)}$ of probability measures on $\cX^n$, let $\cP_n(P^{(1)},P^{(2)})$ be the set of all probability measures on $\cX^n \times \hcX^n$ having $P^{(1)}$ and $P^{(2)}$ as marginals, and define the {\em Wasserstein metric}
\begin{eqnarray*}
&& \overline{\rho}_n(P^{(1)},P^{(2)}) \deq \inf_{P \in \cP_n(P^{(1)},P^{(2)})} \E_P\left\{\rho_n(X^n,Y^n)\right\} \\
&& \quad \equiv \inf_{P \in \cP_n(P^{(1)},P^{(2)})} \int \rho_n(x^n,y^n) dP(x^n,y^n)
\end{eqnarray*}
(See Gray, Neuhoff and Shields \cite{GraNeuShi75} for more details and  applications.) Note that, because $\rho$ is a bounded metric, 
$$
 \int \rho_n(x^n,y^n) dP(x^n,y^n) \le \rho_{\max} \int 1_{\{x^n \neq y^n\}} dP(x^n,y^n)
$$
for all $P \in \cP_n(P^{(1)},P^{(2)})$. Taking the infimum of both sides over all $P \in \cP_n(P^{(1)},P^{(2)})$ and observing that
$$
\dvar (P^{(1)},P^{(2)}) = 2\inf_{P \in \cP_n(P^{(1)},P^{(2)})} \int 1_{\{x^n \neq y^n\}} dP(x^n,y^n),
$$
see, e.g., Section~I.5 of Lindvall \cite{Lin02}, we get the useful bound
\begin{equation}
\overline{\rho}_n(P^{(1)},P^{(2)}) \le \frac{1}{2} \rho_{\max}  \dvar(P^{(1)},P^{(2)}).
\label{eq:wasser_var}
\end{equation}

Now, for each $n$, let $\cM_n$ denote the set of all discrete probability distributions on $\hcX^n$ with finite entropy. That is, $Q \in \cM_n$ if and only if it is concentrated on a finite or a countable set $\{y_i\}_{i \in \cI_Q} \subset \hcX^n$, and
$$
H(Q) \deq -\sum_{i \in \cI_Q} Q(y_i)\log Q(y_i) < \infty.
$$
For every $Q \in \cM_n$, consider the set $\cC(Q)$ of all one-to-one maps $\map{c}{\cI_Q}{\{0,1\}^*}$, such that, for each $c \in \cC(Q)$, the collection $\{c(i)\}_{i \in \cI_Q}$ satisfies the Kraft inequality, and let
$$
\ell_Q \deq \min_{c \in \cC(Q)} \sum_{i \in \cI_Q} \ell(c(i))Q(y_i)
$$
be the minimum expected code length. Since the entropy of $Q$ is finite, there is always a minimizing $c^*_Q$, and the Shannon--Fano bound (see Section~5.4 of Cover and Thomas \cite{CovTho91}) guarantees that $\ell_Q \le H(Q) + 1 < \infty$.

Now, for any $\lambda > 0$, any probability distribution $P$ on $\cX^n$, and any $Q \in \cM_n$, define
$$
L_n(P,Q;\lambda) \deq \overline{\rho}_n(P,Q) + \lambda n^{-1}\ell_Q.
$$
To give an intuitive meaning to $L_n(P,Q;\lambda)$, let $X^n$ and $Y$ be jointly distributed random variables with $X^n \sim P$ and $Y \sim Q$, such that their joint distribution $\overline{P} \in \cP_n(P,Q)$ achieves $\overline{\rho}_n(P,Q)$. Then $L_n(P,Q;\lambda)$ is the expected Lagrangian performance, at Lagrange multiplier $\lambda$, of a {\em stochastic} variable-rate quantizer which encodes each point $x^n \in \cX^n$ as a binary codeword with length $c^*_Q(i)$ and decodes it to $y_i$ in the support of $Q$ with probability $\overline{P}(Y=y_i|X^n=x^n)$.

The following lemma shows that deterministic quantizers are as good as random ones:

\begin{lemma}
\label{lm:ran_det}

Let $L_P(C^n,\lambda)$ be the expected Lagrangian performance of an $n$-block variable rate quantizer operating on $X^n \sim P$, and let $\wh{L}^n_P(\lambda)$ be the expected Lagrangian performance, with respect to $P$, of the best $n$-block variable-rate quantizer. Then
$$
\wh{L}^n_P(\lambda) = \inf_{Q \in \cM_n} L_n(P,Q;\lambda).
$$
\end{lemma}

\begin{proof} Consider any quantizer $C^n = (f,\varphi)$ with $L_P(C^n,\lambda) < \infty$. Let $Q_{C^n}$ be the distribution of $C^n(X^n)$. Clearly, $Q_{C^n} \in \cM_n$, and
\begin{eqnarray*}
L_P(C^n,\lambda) &=& \E\left\{\rho_n(X^n,C^n(X^n))\right\} + \lambda  \E\left\{\ell_n(f(X^n))\right\}\\
&\ge& \overline{\rho}_n(P,Q_{C^n}) + \lambda n^{-1} \ell_{Q_{C^n}} \\
&=& L_n(P,Q_{C^n};\lambda).
\end{eqnarray*}
Hence, $\wh{L}^n_P(\lambda) \ge \inf_{Q \in \cM_n} L_n(P,Q;\lambda)$. To prove the reverse inequality, suppose that $X^n \sim P$ and $Y \sim Q$ achieve $\overline{\rho}_n(P,Q)$ for some $Q \in \cM_n$. Let $\overline{P}$ be their joint distribution. Let $\{y_i\}_{i \in \cI_Q} \subset \hcX^n$ be the support of $Q$, let $\map{c^*_Q}{\cI_Q}{\{0,1\}^*}$ achieve $\ell_Q$, and let $\cS = \{ c^*_Q(i) \}_{i \in \cI_Q}$ be the associated binary code. Define the quantizer $C^n = (f,\varphi)$ by
$$
\varphi(s) = y_i \quad \mbox{if } s = c^*_Q(i)
$$
and
$$
f(x^n) = \argmin_{s \in \cS} \left(\rho_n(x^n,\varphi(s)) + \lambda  \ell_n(s)\right).
$$
Then
$$
L_P(C^n,\lambda) = \E_P \left\{ \min_{s \in \cS} \left(\rho_n(X^n,\varphi(s)) + \lambda \ell_n(s)\right)\right\}.
$$
On the other hand,
\begin{eqnarray*}
&& L_n(P,Q;\lambda) \\
&& \qquad = \E_{\overline{P}} \left\{ \overline{\rho}_n(X^n_1,Y) + \lambda n^{-1} \ell_Q \right\} \\
&& \qquad = \int dP(x^n) \sum_{i \in \cI_Q} \left( \rho_n(x^n,y_i) + \lambda \ell_n(c^*_Q(i))\right)\\
&& \qquad \qquad \qquad \times \overline{P}(Y=y_i|X^n = x^n) \\
&& \qquad \ge \int dP(x^n) \min_{i \in \cI_Q} \left( \rho_n(x^n,y_i) + \lambda  \ell_n(c^*_Q(i))\right) \\
&& \qquad = \int dP(x^n) \min_{s \in \cS} \left( \rho_n(x^n,\varphi(s)) + \lambda \ell_n(s)\right)\\
&& \qquad \equiv L_P(C^n,\lambda),
\end{eqnarray*}
so that $\inf_{Q \in \cM_n} L_n(P,Q;\lambda) \ge \wh{L}^n_P(\lambda)$, and the lemma is proved.
\end{proof}

The following lemma gives a useful upper bound on the Lagrangian mismatch:

\begin{lemma}
\label{lm:mismatch}
Let $P,P'$ be probability distributions on $\cX^n$. Then
$$
\left|\wh{L}^n_P(\lambda) - \wh{L}^n_{P'}(\lambda)\right| \le \frac{1}{2} \rho_{\max} \dvar(P,P').
$$
\end{lemma}

\begin{proof} Suppose $\wh{L}^n_P(\lambda) \ge \wh{L}^n_{P'}(\lambda)$. Let $Q'$ achieve $\inf_{Q \in \cM_n} L_n(P',Q;\lambda)$ (or be arbitrarily close). Then
\begin{eqnarray*}
&& \wh{L}^n_P(\lambda) - \wh{L}^n_{P'}(\lambda) \\
&&\quad \stackrel{{\rm (a)}}{=} \inf_{Q \in \cM_n} L_n(P,Q;\lambda) - \inf_{Q \in \cM_n} L_n(P',Q;\lambda) \\
&&\quad = \inf_{Q \in \cM_n} L_n(P,Q;\lambda) - L_n(P',Q';\lambda) \\
&&\quad \le L_n(P,Q';\lambda) - L_n(P',Q';\lambda) \\
&&\quad \stackrel{{\rm (b)}}{=} \overline{\rho}_n(P,Q') + \lambda n^{-1} \ell_{Q'} -\overline{\rho}_n(P',Q') - \lambda n^{-1} \ell_{Q'} \\
&&\quad =  \overline{\rho}_n(P,Q') - \overline{\rho}_n(P',Q') \\
&&\quad \stackrel{{\rm (c)}}{\le} \overline{\rho}_n(P,P') \\
&&\quad \stackrel{{\rm (d)}}{\le} \frac{1}{2} \rho_{\max} \dvar(P,P'),
\end{eqnarray*}
where in (a) we used Lemma~\ref{lm:ran_det}) in (b) we used the definition of $L_n(\cdot,Q';\lambda)$, in (c) we used the fact that $\overline{\rho}_n$ is a metric and the triangle inequality, and in (d) we used the bound (\ref{eq:wasser_var}).
\end{proof}

Finally, the lemma below shows that, for bounded distortion functions, Lagrange-optimal quantizers have finite codebooks:

\begin{lemma}
\label{lm:finite_codebook}
For positive integers $N,L$, let $\cQ_n(N,L)$ denote the set of all zero-memory variable-rate quantizers with block length $n$, such that for every $C^n \in \cQ_n(N,L)$, the associated binary code $\cS$ of $C^n$ satisfies $|\cS| \le N$ and $\ell(s) \le L$ for every $s \in \cS$. Let $P$ be a probability distribution on $\cX^n$. Then
$$
\wh{L}^n_P(\lambda) = \inf_{C^n \in \cQ_n(N,L)} L_P(C^n,\lambda),
$$
with $N \le 2^{2n\rho_{\max}/\lambda}$ and $L \le 2n\rho_{\max}/\lambda$.
\end{lemma}

\begin{proof} Let $C^n_*$ with encoder $\map{f_*}{\cX^n}{\cS}$ and decoder $\map{\varphi_*}{\cS}{\hcX^n}$ achieve the $n$th-order optimum $\wh{L}^n_P(\lambda)$ for $P$. Let $s_0 \in \cS$ be the shortest binary string in $\cS$, i.e.,
$$
\ell(s_0) = \min_{s \in \cS}\ell(s).
$$
Without loss of generality, we can take $f_*$ as the minimum-distortion encoder, i.e.,
$$
f_*(x^n) = \argmin_{s \in \cS} \left(\rho_n(x^n,\varphi_*(s)) + \lambda \ell_n(s)\right).
$$
Thus, for any $s \in \cS$ and any $x^n \in f^{-1}_*(s)$,
$$
\rho_n(x^n,\varphi_*(s)) + \lambda \ell_n(s) \le \rho_n(x^n,\varphi_*(s_0)) + \lambda \ell_n(s_0).
$$
Hence, $\ell(s) \le n\rho_{\max}/\lambda + \ell(s_0)$ for all $s \in \cS$. Furthermore, $L_P(C^n_*,\lambda) \ge \lambda \E_P\left\{\ell_n(f_*(X^n))\right\} \ge \lambda \ell_n(s_0)$.

Now pick an arbitrary reproduction string $\wh{x}^n_0 \in \hcX^n$, let $\varepsilon$ be the empty binary string (of length zero), and let $C^n_0$ be the zero-rate quantizer with the constant encoder $f_0(x^n) = \varepsilon$ and the decoder $\varphi_0(\varepsilon) = \wh{x}^n_0$. Then $L_P(C^n_0,\lambda) = \E_P\left\{\rho_n(X^n,\wh{x}^n_0)\right\} + \lambda \ell_n(\varepsilon) \le \rho_{\max}$. On the other hand, $L_P(C^n_*,\lambda) \le L_P(C^n_0,\lambda)$. Therefore,
$$
\lambda  \ell_n(s_0) \le L_P(C^n_*,\lambda) \le L_P(C^n_0,\lambda) \le \rho_{\max},
$$
so that $\ell(s_0) \le n\rho_{\max}/\lambda$. Hence,
$$
\ell(s) \le 2n\rho_{\max}/\lambda, \qquad \forall s\in \cS,
$$
Since the strings in $\cS$ must satisfy Kraft's inequality, we have
$$
1 \ge \sum_{s \in \cS} 2^{-\ell(s)} \ge |\cS|2^{-2n\rho_{\max}/\lambda},
$$
which implies that $|\cS| \le 2^{2n\rho_{\max}/\lambda}$.
\end{proof}

\section*{Acknowledgment}

The author would like to thank Andrew R.~Barron, Ioannis Kontoyiannis and Mokshay Madiman for stimulating discussions, and the anonymous reviewers for several useful suggestions that helped improve the paper.

\end{document}